\documentclass[a4paper,10pt,oneside,onecolumn,preprint]{elsarticle}
\usepackage[T1]{fontenc}
\usepackage{amssymb}
\usepackage{latexsym}
\usepackage{amsfonts}
\usepackage{amsmath}
\usepackage{amsthm}
\usepackage{caption}
\usepackage[utf8]{inputenc}
\usepackage{hyperref}
\usepackage{tikz}
\usepackage{soul}
\usepackage{epsfig}
\usepackage{amsfonts}
\usepackage{framed} 
\usepackage{algorithm}
\usepackage{algpseudocode}
\usepackage{geometry}
\usepackage{graphics}
\usepackage{array}
\usepackage{resizegather}
\usepackage{yfonts}
\usepackage{ifthen}
\usepackage{centernot}
\usepackage[justification=centering]{caption}
\usepackage{mathrsfs}
\usepackage{multirow}
\usepackage{adjustbox}
\usepackage{mdframed}
\usepackage{hyperref}
\usepackage{thmtools,thm-restate}
\usetikzlibrary{matrix}

\newcommand{\footref}[1]{%
    $^{\ref{#1}}$%
}

\algrenewcommand\algorithmicrequire{\textbf{Input:}}
\algrenewcommand\algorithmicensure{\textbf{Output:}}

\def\myalgorithm#1{\textsf{#1}}
\def\closedinterval#1#2{\left\{#1,\ldots,#2\right\}}

\newtheorem{theorem}{\textbf{Theorem}}
{\bfseries}{\itshape}{\rmfamily}
\newtheorem{proposition}[theorem]{\textbf{Proposition}}
\newtheorem{lemma}[theorem]{\textbf{Lemma}}

\newtheorem{remark}[theorem]{\textbf{Remark}}
\newtheorem{corollary}[theorem]{\textbf{Corollary}}
\newtheorem{definition}[theorem]{\textbf{Definition}}
\newtheorem{notation}[theorem]{Notation}
\newtheorem{example}[theorem]{Example}
\DeclareMathOperator{\spn}{Span}

\DeclareMathOperator{\stb}{Stab}
\DeclareMathOperator{\rk}{rk}

\DeclareMathOperator{\GL}{GL}

\DeclareMathOperator{\RP}{RP}
\DeclareMathOperator{\comp}{\mathop{\circ}}
\newcommand\AL[1]{
	\mathscr{S}_{{#1}}
	}
\newcommand\Bl[2]{\mathcal{L}(K^{#1},K^{#2};K)}
\newcommand\VecSp[1]{#1}
\newcommand\GLGL[2]{\GL({K^#1})\times\GL({K^#2})}
\newcommand\Near[2]{
	\ifthenelse{\equal{#2}{}}{
		\mathcal{V}_1(#1)}{
		\mathcal{V}_{#2}(#1)
		}
	}
\newcommand\Cover[2]{
	\AL{#1}(\spn(#2))
	}
\newcommand\QCover[2]{
	\tilde{\mathscr{S}}_{#1}(\{#2\})
	}
\newcommand\QQCover[2]{
	\tilde{\mathscr{S}}_{#1}(#2)
	}
\newcommand\CCover[2]{
	\AL{#1}(#2)
	}
\newcommand\ALQ[2]{
	\ifthenelse{\equal{#2}{}}{\quotient{\Cover{r}{#1}}{\stb(\VecSp{#1})}}{\quotient{\Cover{#2}{#1}}{\stb(\VecSp{#1})}}
	}
\newcommand\ShProdMap[1]{\VecSp{T}}
\newcommand\Trans[1]{{#1}^{\mathrm{T}}}

\makeatletter
\def\mproof#1{%
    \trivlist
    \item[%
        \hskip 10\p@
        \hskip \labelsep
        {\sc #1.}%
    ]
    \ignorespaces
}
\def\mqed{%
    \unskip
    \kern 10\p@
    \hfill
    \begingroup
        \unitlength\p@
        \linethickness{.4\p@}%
        \framebox(6,6){}%
    \endgroup
    \global\@qededtrue
}
\def\bigslant#1#2{%
    \raise0em\hbox{$#1$}\lower0.2em\hbox{/}\raisebox{-0.5em}[0.1pt]{\footnotesize$#2$}%
}
\def\medslant#1#2{%
    \hbox{\footnotesize$#1$}\lower0.2em\hbox{/}\raisebox{-0.5em}[0.1pt]{\tiny$#2$}%
}

{{\hfill{$\square$}\bgroup\parfillskip=0pt\par\egroup}}

\newcommand{\quotient}[2]{%
\hbox{\(#1\)}\kern-.0em%
/%
\kern-.1em\lower.35ex\hbox{\(#2\)}}

\makeatletter
\def\@fixme[#1]#2{\oldmarginpar{\smash{\hbox to 0pt{\hss\tikz\node[rotate=-90,rectangle,draw=red,fill=red!15!white,opacity=0.4,very thick,rounded
corners](X){\parbox[c]{#1}{\tiny #2}};\hss}}}}%
\def\fixme{\@ifnextchar[{\@fixme}{\@fixme[8em]}}
\makeatother

\makeatletter
\let\olog\log

\def\mlog#1{\olog^{(#1)}}

\def\@zlog{\ifnum\@tempcnta=1\olog\else\mlog{\the\@tempcnta}\fi}
\def\@loginner{\ifx\reserved@a\log\def\next\log{\@log}\else\let\next\@zlog\fi\next}
\def\@log{\advance\@tempcnta1\futurelet\reserved@a\@loginner}
\DeclareRobustCommand{\log}{\@tempcnta0\@log}
\makeatother

\begin{document}
\title{Improved method for finding optimal formulae for bilinear maps in a finite field}

\author{Svyatoslav Covanov}
\ead{svyatoslav.covanov@inria.fr}
\address{
Universit\'e de Lorraine, LORIA, UMR 7503, Vandoeuvre-l\`es-Nancy, F-54506, France}
\address{
	Inria, Villers-l\`es-Nancy, F-54600, France}
\address{
	CNRS, LORIA, UMR 7503, Vandoeuvre-l\`es-Nancy, F-54506, France
}

\begin{abstract}
	In 2012, Barbulescu, Detrey, Estibals and Zimmermann proposed a new framework to exhaustively search for
	optimal	formulae for evaluating bilinear maps over finite fields, such as Strassen or Karatsuba formulae.
	The main contribution of this work is a new criterion to aggressively prune useless branches in the 
	exhaustive search, thus leading	to the computation of new optimal formulae. We apply in particular our approach to the short product modulo $X^5$
	and the circulant product modulo $(X^5-1)$.
	Moreover, we are able to prove that there is essentially only one optimal decomposition of the product of 
	$3\times 2$ by $2\times 3$ matrices
	up to the action of some group of automorphisms.
\end{abstract}
\maketitle

\section{Introduction}
Finding optimal formulae for computing bilinear maps is a problem of algebraic complexity
theory~\cite{Brgisser:2010:ACT:1965416,BROCKETT1978207,strassen,doi:10.1137/0208037}, initiated by the discoveries of
Karatsuba and Ofman~\cite{karatsuba-en} and Strassen~\cite{strassen}.
It consists in determining
almost optimal algorithms for important problems of complexity theory, among which the well studied complexity of
matrix multiplication~\cite{strassen,Pan:1978:SAO:1382432.1382584,COPPERSMITH1990251,LeGall:2014:PTF:2608628.2608664}
and the complexity of polynomial multiplication~\cite{karatsuba-en,1963-toom,1971-scho,DBLP:journals/corr/HarveyHL14}.

As far as polynomial multiplication is concerned, the first improvement over the schoolbook method
came from Karatsuba and Ofman~\cite{karatsuba-en} in 1962, who proposed a decomposition of the bilinear map
associated to the product of two polynomials of degree $1$
$$A = a_0 + a_1 X \text{ and } B = b_0 + b_1 X.$$
Using the schoolbook algorithm, computing the product $A\cdot B$
requires ${4}$ multiplications over the coefficient ring: $a_0b_0$, $a_1b_0$, $a_0b_1$, $a_1b_1$.
With the algorithm proposed by Karatsuba, the coefficients of the product $A\cdot B$
can be retrieved from the computation of the ${3}$ following multiplications: $a_0b_0$, $(a_0+a_1)(b_0+b_1)$, $a_1b_1$.
In particular, Karatsuba's algorithm can be applied recursively to improve the binary complexity of the multiplication of two $n$-bit integers:
instead of $O(n^2)$ with the naive schoolbook algorithm, we obtain $O(n^{\log_2 3})$.

In 1969, Strassen~\cite{strassen} proposed formulae improving on the cost of the product of two $2\times 2$ matrices.
When applied recursively on large matrices, this leads to a binary complexity of $O(n^{\log_2 7})$
instead of $O(n^{\log_2 8}) = O(n^3)$. Smirnov describes in~\cite{Smirnov2013} practical algorithms for matrices of higher
dimensions. One can notice that, most of the time, optimal algorithms for matrix multiplication are unknown.
For example, it is possible compute the product of $3\times 3$ matrices over $\mathbb{C}$ with $23$ multiplications~\cite{MR0395320},
but the best known lower bound is still $19$~\cite{Blaser200343}.

\paragraph{\textbf{State of the art}}
An obstacle to finding optimal formulae is the fact that the decomposition of bilinear maps is known to be NP-hard~\cite{Hastad1990644}.
In terms of method, the least-squares method seems to be one of the most popular~\cite{Smirnov2013}.
Another way to decompose a bilinear map consists in using ingredients from geometry~\cite{Bernardi201351} and to find a generalization
of the decomposition of singular value decomposition for matrices to general tensors.
However, these methods
are essentially used over an algebraically closed field $K$ (e.g. $K=\mathbb{C}$) and are not meant to produce all the possible
decompositions for a bilinear map.
In our context, we are looking for a method computing optimal formulae (for the bilinear rank) over a finite field $K$.
These formulae can be used for the same bilinear map over any extension of $K$.
Thus, they can be used in the context of the asymptotic multiplication of polynomials over a finite field for example.
Furthermore, for a set of formulae of $\mathbb{Q}$, we can deduce formulae over a finite field $K$.
It should be possible, given all the optimal formulae over $K$ to obtain formulae over $\mathbb{Q}$.
In other terms, finding optimal formulae over finite fields can be used to improve on the multiplication algorithms
over larger fields.

Montgomery proposed in~\cite{1388200} an algorithm to compute such a decomposition for the particular case
of polynomials of small degree over a finite field. The author takes advantage
of the fact that the number of possible formulae is always finite on a finite field. He obtains new formulae for the multiplication
of polynomials of degree $4$, $5$ and $6$ over $\mathbb{F}_2$. In~\cite{Oseledets20082052}, Oseledets proposes
a heuristic approach
to solve the bilinear rank problem for the polynomial product over $\mathbb{F}_2$.
Later, Barbulescu et al.\ proposed in~\cite{Barbulescu2012}
a unified framework, extending the idea proposed by Oseledets.
This allows the authors to compute
the bilinear rank of different applications, such as the short product or the middle
product over a finite field. Their algorithm allows one to generate all the possible rank decompositions
of any bilinear map over a finite field. We extend this work in the current article.
\par

\paragraph{\textbf{Contributions}}
The work presented is an improvement to the algorithm introduced in~\cite{Barbulescu2012}, allowing one 
to increase the family of bilinear maps over a finite field for which we are able to compute all the optimal formulae.
Our algorithm relies on the automorphism group stabilizing a bilinear map,
and on the notion of ``stem'' of a vector space associated to such a bilinear map.
The main theorem of this work is Theorem~\ref{thm:mainthm} and it states that Algorithm~\ref{alg:mainalg}
is able to find all decompositions of a bilinear map over a finite field.
It can be used for proving lower bounds on the rank of a bilinear map and it has applications for improving
upper bounds on the Chudnovsky-Chudnovsky algorithms~\cite{CHUDNOVSKY1988285,Randriambololona2012489,Rambaud2015}.
Specifically, we compute all the decompositions
for the short product of polynomials $P$ and $Q$ modulo $X^5$
and the product of $3\times 2$ by $2\times 3$ matrices. The latter problem was out of reach with
the method used in~\cite{Barbulescu2012}. We prove, in particular, that the set of possible decompositions for this matrix product
is essentially unique, up to the action of the automorphism group.
It is difficult to propose a complexity analysis showing the impact of our method, since
it takes into account intrinsic properties of the bilinear maps that are considered.
\par

\paragraph{\textbf{Roadmap}}
This article is organized as follows. In Section~\ref{sec:bilrank}, we present the theoretical tools
and the framework for this article, corresponding to the framework introduced in~\cite{Barbulescu2012}.
In Section~\ref{sec:bdezstab}, we present, {with kind permission of the authors}, unpublished improvements~\cite{KaraNWithOrbits} taking into
account the symmetries of bilinear maps.
In Section~\ref{sec:algstruct}, we describe the algebraic structure of specific bilinear maps. This section can be skipped
on a first read, because it is only required in proofs of the following section.
In Section~\ref{sec:newalgo}, we describe the theoretical aspect of our main contribution, which relies on the construction
of coverings, and illustrate it with the examples of the short product and the matrix product.
We discuss specific algorithmic aspects in Section~\ref{sec:computetopo}:
this part is quite technical and can be skipped on a first read.
Finally, experimental timings 
are given in Section~\ref{sec:expresults}.
\par

\section{Preliminaries}
\label{sec:bilrank}
We present in this section the definition of the mathematical objects that we manipulate in this work
and we define the bilinear rank. We choose the characterization given by de Groote~\cite{deGroote:1987:LCB:23720}
or B\"urgisser et al.~\cite[Ch.\ 14]{Brgisser:2010:ACT:1965416}. In particular, we introduce here the framework of~\cite{Barbulescu2012}
and the underlying linear algebra problem.

\subsection{Problem statement}
Let ${K}$ be a field. Given a bilinear map $\mathbf{\Phi} :
{K}^m \times {K}^n  \rightarrow {K}^\ell$, the bilinear rank problem
consists in finding the minimal number of multiplications between scalars
used for evaluating $\mathbf\Phi$. The set $\mathcal{L}(K^m,K^n;K^\ell)$ denotes
the set of bilinear maps from $K^m\times K^n$ to $K^\ell$. Any bilinear map $\mathbf{\Phi}$
from $K^m\times K^n$ to $K^\ell$ can be seen as an element of $\Bl{m}{n}^\ell$,
whose coordinates are the bilinear forms $(\Phi_h)_{0\leq h < \ell }$.

\begin{example}[Multiplication of linear polynomials]
Let $A = a_0 + a_1 X$ and $B = b_0 + b_1 X$ be two polynomials over $K$. The product $A\cdot B$ is associated to the bilinear map
$\mathbf{\Phi}$ taking as input the vectors $\mathbf{a} = (a_0,a_1)$ and $\mathbf{b} = (b_0,b_1)$ such that
$$\mathbf\Phi =\begin{pmatrix} \Phi_0 \\ \Phi_1 \\ \Phi_2\end{pmatrix} : (\mathbf{a},\mathbf{b}) \mapsto 
	\begin{pmatrix} a_0 b_0\\ a_0 b_1 + a_1 b_0\\ a_1 b_1\end{pmatrix}.$$

	Denoting by $\phi_0$, $\phi_1$, $\phi_2$ and $\phi_3$ the bilinear forms
	$(\mathbf{a},\mathbf{b}) \mapsto a_0 b_0$, $(\mathbf{a},\mathbf{b}) \mapsto a_0 b_1$, 
	$(\mathbf{a},\mathbf{b}) \mapsto a_1 b_0$ and $(\mathbf{a},\mathbf{b}) \mapsto a_1 b_1$,
	respectively,
	we have
	$$\mathbf{\Phi} =
		\phi_0\cdot \begin{pmatrix}
	1\\ 0\\0
	\end{pmatrix} + \phi_1\cdot \begin{pmatrix}
	0\\ 1\\0
	\end{pmatrix}+\phi_2\cdot \begin{pmatrix}
	0\\ 1\\0
			\end{pmatrix} + \phi_3 \cdot \begin{pmatrix} 0 \\ 0 \\1\end{pmatrix},$$
	which corresponds to the schoolbook algorithm.

Let $\psi$ be an element of $\Bl{2}{2}$ such that $\psi : (\mathbf{a},\mathbf{b}) \mapsto (a_0 + a_1) (b_0 +b_1)$.
Then, since $\phi_1+\phi_2 = \psi - \phi_0 - \phi_3$, we can rewrite $\mathbf{\Phi}$ as

$$\mathbf{\Phi} =
	\phi_0\cdot \begin{pmatrix}
1\\ -1\\0
\end{pmatrix} + \psi\cdot \begin{pmatrix}
0\\ 1\\0
\end{pmatrix}+\phi_2\cdot \begin{pmatrix}
0\\ -1\\1
\end{pmatrix}.$$

The bilinear forms $\phi_0$, $\psi$ and $\phi_2$ each correspond to exactly one multiplication over $K$.
This decomposition corresponds to the Karatsuba algorithm.
Thus, we can deduce that the bilinear rank of $\mathbf\Phi$ is at most $3$.
Actually, one can show that the bilinear rank of $\mathbf\Phi$ is
equal to $3$.
\end{example}

Formally, a bilinear form $\phi \in \Bl{m}{n}$ is said to have rank one if there exist two linear forms $\alpha \in \mathcal{L}(K^m;K)$ and
$\beta \in \mathcal{L}(K^n;K)$ such that $\phi(\mathbf{a},\mathbf{b}) = \alpha(\mathbf{a}) \cdot \beta(\mathbf{b})$.
For $i\in \{0,\ldots,m-1\}$ and $j \in \{0,\ldots,n-1\}$, we denote by $e_{i,j}$ the bilinear forms $e_{i,j} : (\mathbf{a},\mathbf{b}) \mapsto a_ib_j$.
The $e_{i,j}$'s have rank one and form the canonical basis of $\Bl{m}{n}$. This implies that
any bilinear form can be expressed as a linear combination of bilinear forms of rank one.

\begin{definition}[Bilinear rank]
	\label{def:bilrank}
	The rank of a bilinear form $\Phi$, denoted by $\rk(\Phi)$, is defined as the minimal number of bilinear forms $\phi_t$ of rank one
	such that $\Phi$ is a linear combination of the $\phi_t$'s. Then, a family $(\phi_t)_t$ of cardinality $\rk(\Phi)$
	is said to be an optimal decomposition of $\Phi$.

	We extend this definition to bilinear maps $\mathbf{\Phi}\in\Bl{m}{n}^\ell$: the rank $r$ of $\mathbf{\Phi}$ is the cardinality
	of a minimal set of bilinear forms $(\phi_t)_{0\leq t < r}$ of rank one for which there exist
	vectors $\mathbf{c}_t \in K^\ell$ such that 
	$$\mathbf{\Phi} = \sum_{0\leq t < r} \phi_t \cdot\mathbf{c}_t.$$
\end{definition}

We have a matrix equivalent of Definition~\ref{def:bilrank}.
Indeed, for $\Phi \in \Bl{m}{n}$, there exists a matrix $M\in \mathcal{M}_{m,n}(K)$
such that $\Phi(\mathbf{a},\mathbf{b}) = \Trans{\mathbf{a}}\cdot M \cdot \mathbf{b}$
for $\mathbf{a} \in K^m$ and $\mathbf{b} \in K^n$. In this situation, the usual matrix rank of $M$ is equal to the rank
of $\Phi$ defined as above.
Let $\mathbf{\Phi}=(\Phi_0,\ldots,\Phi_{\ell-1})$ be a bilinear map of rank $r$, for which each $\Phi_h$ for
$0\leq h < \ell$
is represented by $M_h \in \mathcal{M}_{m,n}(K)$. Consequently, there exists a set of $r$ matrices $N_t \in \mathcal{M}_{m,n}(K)$
of rank one such that
$$\forall h \in \closedinterval{0}{\ell-1},\ M_h \in \spn(\{N_0,\ldots,N_{r-1}\}).$$

\begin{example}[Short product of polynomials of degree $2$]
We describe in this example the matrices associated to the short product of two polynomials of degree $2$.

Let $A$ and $B$ be the polynomials $A = a_0 +a_1X+a_2X^2$ and $B = b_0 + b_1 X+b_2X^2$.
We denote by $C$ the polynomial $A\cdot B \bmod X^3$:
$$C = a_0 b_0 + (a_0b_1 + a_1b_0)X + (a_0b_2+a_1b_1+a_2b_0)X^2.$$
We consider $A$ and $B$ as vectors of $K^3$ denoted by $\mathbf{a}$
and $\mathbf{b}$, respectively. Let $\Phi_0$, $\Phi_1$ and $\Phi_2$ be bilinear forms defined as
\begin{equation*}
\begin{split}
\Phi_0 : \ & (\mathbf{a},\mathbf{b}) \mapsto a_{0}b_{0}, \\
\Phi_1 : \ & (\mathbf{a},\mathbf{b}) \mapsto a_{0}b_{1} + a_{1}b_{0}, \\
\Phi_2 : \ & (\mathbf{a},\mathbf{b}) \mapsto a_{0}b_{2} + a_{1}b_{1}+a_2b_0. \\
\end{split}
\end{equation*}
	In order to represent the corresponding matrices,
	we use the canonical basis for $\Bl{3}{3}$, i.e. the bilinear forms $e_{i,j}$ satisfying
$e_{i,j} : (\mathbf{a},\mathbf{b}) \mapsto a_ib_{j}$, for $0\leq i,j<3$. Then, the matrices $M_h$ associated to $\Phi_h$ are
\begin{equation*}
\begin{split}
M_0 = \begin{pmatrix}
	1 & 0 & 0\\
	0 & 0 & 0\\
	0 & 0 & 0\\
\end{pmatrix},
M_1 = \begin{pmatrix}
	0 & 1 & 0\\
	1 & 0 & 0\\
	0 & 0 & 0\\
\end{pmatrix},
M_2 = \begin{pmatrix}
	0 & 0 & 1\\
	0 & 1 & 0\\
	1 & 0 & 0\\
\end{pmatrix}.
\end{split}
\end{equation*}
\end{example}

\subsection{A linear algebra problem}
\label{sec:delta}

The approach of~\cite{Barbulescu2012} consists in computing the rank of a bilinear map $\mathbf{\Phi}=(\Phi_0,\ldots,\Phi_{\ell-1})$
by considering $T = \spn(\{\Phi_0,\ldots,\Phi_{\ell-1}\})$, which is a subspace of $\Bl{m}{n}$.
Indeed, finding formulas for computing the $\Phi_t$'s is equivalent to finding a family of rank-one bilinear forms
generating $\VecSp{T}$.
Thus, we need to extend the definition
of the rank to subspaces of $\Bl{m}{n}$.

\begin{notation}
	\label{not:delta}
	For $\VecSp{T}$ a subspace of $\Bl{m}{n}$, we denote by $\CCover{m,n,r}{\VecSp{T}}$ the set of subspaces $\VecSp{V}\subset\Bl{m}{n}$ spanned by a free family of
rank-one bilinear forms of size $r$ such that
	$T\subset \VecSp{V}$.
	
	When $\VecSp{T} = \spn(\emptyset)$, $\CCover{m,n,r}{\VecSp{T}}$ is the set of subspaces $\VecSp{V}\in\Bl{m}{n}$ spanned by a free family of rank-one bilinear forms of size $r$ and we denote it simply by $\AL{m,n,r}$.

When $m$ and $n$ are clear from the context, these sets are simply denoted by $\CCover{r}{\VecSp{T}}$ and $\AL{r}$.
\end{notation}
We use Notation~\ref{not:delta} to define the rank of a subspace $\VecSp{T} \in \Bl{m}{n}$ in Definition~\ref{def:ranksp}.

\begin{definition}[Rank of a subspace of $\Bl{m}{n}$]
	\label{def:ranksp}
	Let $\VecSp{T}$ be a subspace of $\Bl{m}{n}$. We denote by
	$\rk(\VecSp{T})$ the smallest $r$ such that $\CCover{r}{\VecSp{T}} \neq \emptyset$.
	The set $\CCover{\rk(\VecSp{T})}{\VecSp{T}}$ is the said to be the set of optimal decompositions of $\VecSp{T}$.
\end{definition}
We observe that $\rk(\VecSp{T}) \geq \dim(\VecSp{T})$.

Let $\mathbf{\Phi} = (\Phi_0,\ldots,\Phi_{\ell-1}) \in \Bl{m}{n}^\ell$ and
$\VecSp{T} = \spn(\{\Phi_0,\ldots,\Phi_{\ell-1}\})\subset \Bl{m}{n}$.
Decomposing a bilinear map $\mathbf{\Phi} \in \Bl{m}{n}^{\ell}$ into linear combination of
 $r$ rank-one bilinear forms is equivalent to computing
$\CCover{r}{\VecSp{T}}$.
Our approach focuses on the latter point of view,
which is also the point of view taken by Algorithm~\cite[Alg. 1]{Barbulescu2012}.

\paragraph{\textbf{General strategy for computing the bilinear rank}}
Taking into account the formalism proposed in Section~\ref{sec:delta}, the algorithmic strategy we use
to compute the bilinear rank of a bilinear map is stated as follows.
\begin{itemize}
	\item Let $\VecSp{T} = \spn(\{\Phi_0,\ldots,\Phi_{\ell-1}\}) \subset \Bl{m}{n}$ of dimension $\ell$;
\item start with the known lower bound $r = \ell$ on the bilinear rank;
\item compute $\CCover{r}{\VecSp{T}}$;
\item if $\CCover{r}{\VecSp{T}} = \emptyset$, increment $r$ and return to the previous step;
\item if $\CCover{r}{\VecSp{T}} \neq \emptyset$, $r$ is the bilinear rank and $\CCover{r}{\VecSp{T}}$ the set of optimal decompositions.
\end{itemize}
\par
\subsection{The BDEZ Algorithm (Barbulescu, Detrey, Estibals, Zimmermann)}
\label{sec:exhaustivesearch}
We describe in this section Algorithm~\cite[Alg. 1]{Barbulescu2012}, which is a recursive method to solve
the bilinear rank problem for a bilinear map $\mathbf{\Phi} = (\Phi_0,\ldots,\Phi_{\ell-1})$ over a finite field.
As described above, this is essentially equivalent to
computing $\CCover{r}{\VecSp{T}}$ for $\VecSp{T} = \spn(\{\Phi_0,\ldots,\Phi_{\ell-1}\})$ of dimension $\ell$.

In order to get all the vector spaces
$\VecSp{V} \in \AL{r}$ such that $\VecSp{T} \subset \VecSp{V}$, we compute
the vector spaces $\VecSp{W} \in \AL{r-\ell}$  such that $\VecSp{T}\oplus \VecSp{W} \in \AL{r}$.
In other terms, instead of enumerating all the elements of $\AL{r}$, we rather enumerate complementary subspaces of $\VecSp{T}$
in $\AL{r-\ell}$.
This restriction can be done thanks to
Proposition~\cite[Prop. 1]{Barbulescu2012}, reformulated as Proposition~\ref{prop:baseincomplete} using the formalism
of Section~\ref{sec:delta}.

\begin{proposition}
	\label{prop:baseincomplete}
	Let $\VecSp{T}$ be a subspace of dimension $\ell$ of $\Bl{m}{n}$, let $r\geq \ell$ be an integer.
	For any $\VecSp{V} \in \CCover{r}{\VecSp{T}}$, there exists $\VecSp{W} \in \AL{r-\ell}$ such that $\VecSp{T}\oplus \VecSp{W} = \VecSp{V}$.
\end{proposition}
\begin{proof}
	Let $\mathcal{B}$ be a basis of $V$ composed of rank-one matrices. We define inductively a sequence of subspaces
	$(\VecSp{W}_t)_{0\leq t\leq r-\ell}$, such that for any $t$ we have
	$\VecSp{W}_t\in\AL{t}$, as follows.
	\begin{itemize}
		\item The set $\VecSp{W}_0$ is the null subspace and satisfies
		$\VecSp{T}\oplus \VecSp{W}_{0} \subset V$ and $\dim{\VecSp{T}\oplus \VecSp{W}_{0}} = \ell$.

		\item For $t \in \closedinterval{1}{r-\ell}$, assuming that $\VecSp{T}\oplus \VecSp{W}_{t-1}
		\subset V$ and
			$\dim{(\VecSp{T}\oplus \VecSp{W}_{t-1})} = \ell+t-1$, there exists $\Phi\in \mathcal{B}$
		such that $\Phi\not\in \VecSp{T}\oplus \VecSp{W}_{t-1}$
		(otherwise $\VecSp{T}\oplus \VecSp{W}_{t-1} =V$ and
		$\dim{V} \leq r-1$, which is a contradiction).
			Then, we define $\VecSp{W}_t$ as $\VecSp{W}_t = \VecSp{W}_{t-1}\oplus\spn(\{\Phi\})$.
		The subspace $\VecSp{W}_t$ satisfies
			$\VecSp{T}\oplus \VecSp{W}_{t} \subset V$, $\dim{(\VecSp{T}\oplus \VecSp{W}_{t})} = \ell+t$
		and $\VecSp{W}_t \in \AL{t}$.
	\end{itemize}
	Taking $\VecSp{W} = \VecSp{W}_{r-\ell}$, Proposition~\ref{prop:baseincomplete} is proved.
\end{proof}
We denote by $\mathcal{G}$ the set of rank-one bilinear forms up to a multiplicative factor, isomorphic
to $\AL{m,n,1}$.
In a finite field, $\mathcal{G}$ is
a finite set of cardinality ${(\#K^m-1) (\#K^n-1)}/{(\#K-1)^2}$.
Algorithm~\myalgorithm{BDEZ} requires a test to determine whether, for
$\VecSp{V}\in\Bl{m}{n}$ of dimension $r$, we have $\VecSp{V} \in \AL{r}$: we denote
by~\myalgorithm{HasRankOneBasis} this test. A naive method to perform this test is described
in Algorithm~\ref{alg:basictest}. We could think of other methods based on solving bilinear systems,
but it does not seem efficient in our applications. However, an optimized version of this algorithm
is used for particular bilinear maps (such as product of $2\times 3$ by $3\times 2$ matrices, for example).

\begin{algorithm}[H]                      
	\caption{\myalgorithm{HasRankOneBasis} (naive method)}
\label{alg:basictest}

\begin{algorithmic}[1]                    
	\Require{$\VecSp{V}$ subspace of $\Bl{m}{n}$}
	\Ensure{Boolean indicating whether $V \in \AL{\dim(\VecSp{V})}$}
	\State $\mathcal{H} \leftarrow \mathcal{G}\cap \VecSp{V}$ \Comment $\#\mathcal{G}$ membership tests (Gaussian elimination)
	\If{$\dim(\spn(\mathcal{H})) = \dim(\VecSp{V})$}
	\State \Return \textbf{true}
	\Else
	\State \Return \textbf{false}
	\EndIf
\end{algorithmic}
\end{algorithm}

Algorithm~\myalgorithm{BDEZ} can be described as a  
recursive optimized version of the backtracking method constructing all the sets of cardinality
$r-\ell$ of independent bilinear forms of rank one.
The input of the first call to~\myalgorithm{BDEZ} is: a target subspace $\VecSp{T}$ of dimension $\ell$ 
and an integer $r$ ($r$ is a lower bound on the rank of $\VecSp{T}$, as explained at the end
of Section~\ref{sec:delta}).
\begin{algorithm}[H]                      
	\caption{\myalgorithm{BDEZ}}
\label{alg:recapproach}

\begin{algorithmic}[1]                    
	\Require{$\VecSp{T}\subset \Bl{m}{n}$ of dimension $\ell$, an integer $r$}
	\Ensure{$\CCover{r}{\VecSp{T}}$}
	\Function{\myalgorithm{ExpandSubspace}}{$\VecSp{V},\mathcal{H},d,r$}
	\If{$d = r$ and $\dim \VecSp{V} = r$ and \myalgorithm{HasRankOneBasis}($\VecSp{V}$)}
	\State \Return $\{\VecSp{V}\}$
	\Else
	\State $\mathcal{S} \leftarrow \emptyset$
	\For{$i \in \closedinterval{0}{\#\mathcal{H}-1}$} \Comment $\mathcal{H} = \{\phi_i | \ i \in [0,\#\mathcal{H}-1]\}$
	\State $\mathcal{H}' \leftarrow \{\phi_{i+1},\ldots,\phi_{\#\mathcal{H}-1}\} \mod \phi_i$
	\Comment Gauss reduction modulo $\phi_i$
	\label{enumrec:hmodw}
	\State $\mathcal{S} \leftarrow$ $\mathcal{S}\ \cup\ $\Call{\myalgorithm{ExpandSubspace}}{$\VecSp{V}\oplus \spn(\{\phi\}_i), \mathcal{H}', d+1, r$}
	\EndFor
	\State \Return $\mathcal{S}$
	\EndIf
	\EndFunction
	\State \Return \Call{\myalgorithm{ExpandSubspace}}{$\VecSp{T}, \mathcal{G} \bmod \VecSp{T}, \ell, r$}
	\Comment Gauss reduction of $\mathcal{G}$ modulo a basis of $\VecSp{T}$
	\label{bdez:calltoexp}
\end{algorithmic}
\end{algorithm}

Algorithm~\myalgorithm{BDEZ} takes into account, on Line~\ref{enumrec:hmodw}, the equivalence relation ``modulo $\VecSp{V}$'':
two distinct elements $\phi$ and $\phi'$ of $\mathcal{H}$ may be such that $V+\spn(\{\phi\}) = V+\spn(\{\phi'\})$.
Reducing each element of $\mathcal{H}$ against $\VecSp{V}$ (via Gauss reduction) allows us
to consider a single representative for each such equivalence class modulo $\VecSp{V}$.
A similar reduction is performed on Line~\ref{bdez:calltoexp} to compute $\mathcal{G} \bmod \VecSp{T}$.

The recursive calls of this algorithm can be represented by a tree in which each node at depth $r-\ell$ corresponds to a vector space
$\VecSp{T}\oplus W_{u_1,u_2,\ldots ,u_{r-\ell}}$ of dimension $r$ generated by a basis of $\VecSp{T}$ and rank-one matrices
$\phi_{u_1},\phi_{u_2},\ldots,\phi_{u_{r-\ell}}$.
For example, assuming that the initial set of rank-one bilinear forms is $\mathcal{G} = \{\phi_0,\phi_1,\phi_2,\phi_3\}$ and
ignoring the reductions computed on Line~\ref{enumrec:hmodw},
we would obtain  generically, for $r-\ell=3$, the
tree given in Figure~\ref{Graph:karan1}.

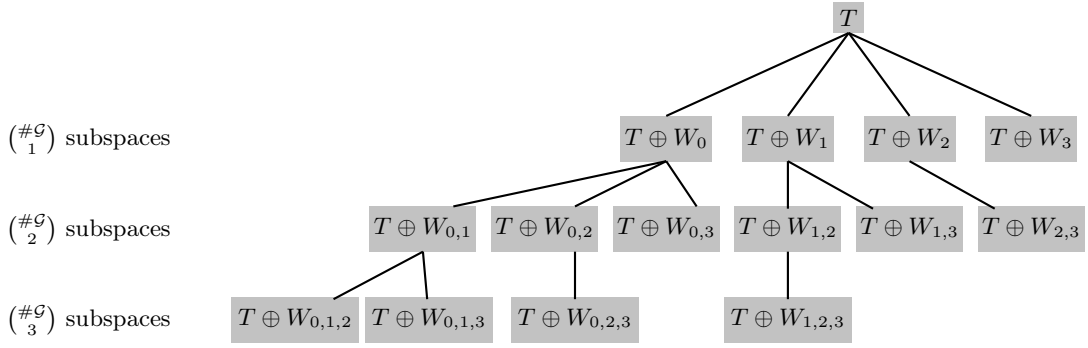
\begin{figure}[H]
\begin{center}
    \begin{tikzpicture}[scale=0.40,font=\small]
		\node (ND1) at (-25,-3.5) {$\binom{\#\mathcal{G}}{1}$ subspaces};
		\node (ND1) at (-25,-6.5) {$\binom{\#\mathcal{G}}{2}$ subspaces};
		\node (ND1) at (-25,-9.5) {$\binom{\#\mathcal{G}}{3}$ subspaces};
	    \fill[black!20!lightgray!60] (-0.5,0)
        rectangle (0.5,1);
        \node at (0,0.5) {$\VecSp{T}$};
        \foreach\i in {0,1,2,3} {
    	    \fill[black!20!lightgray!60] (4*\i-6.0-1.50,-2.75)
	        rectangle (4*\i-6.0+1.50,-4.25);
		\node (N\i) at (4*\i-6.0,-3.5) {$\VecSp{T}\oplus W_{\i}$};
	        \draw[thick] (0,0) -- (4*\i-6.0,-2.75);
        }
        \foreach\i in {1,2,3} {
       	    \fill[black!20!lightgray!60] (4*\i-18.0-1.75,-5.75)
	        rectangle (4*\i-18.0+1.75,-7.25);
	        \node (N0\i) at (4*\i-18.0,-6.5) {$\VecSp{T}\oplus W_{0,\i}$};
            \draw[thick] (-6.0,-4.25) -- (4*\i-17.0,-5.75);
        }
        \foreach\i in {2,3} {
       	    \fill[black!20!lightgray!60] (7.0+4*\i-17.0-1.75,-5.75)
	        rectangle (7.0+4*\i-17.0+1.75,-7.25);
	        \node (N1\i) at (7.0+4*\i-17.0,-6.5) {$\VecSp{T}\oplus W_{1,\i}$};
            \draw[thick] (4-6.0,-4.25) -- (N1\i);
        }
        \fill[black!20!lightgray!60] (6.0-1.75,-5.75)
        rectangle (6.0+1.75,-7.25);
        \node (N23) at (6.0,-6.5) {$\VecSp{T}\oplus W_{2,3}$};
        \draw[thick] (2.0,-4.25) -- (N23);
        \foreach\i in {2,3} {
        	\fill[black!20!lightgray!60] (4*\i*1.1-19.0-8.0-2.1,-8.75)
	        rectangle (4*\i*1.1-19.0-8.0+2.1,-10.25);
	        \node (N01\i) at (4*\i*1.1-19.0-8.0,-9.5) {$\VecSp{T}\oplus W_{0,1,\i}$};
            \draw[thick] (4-18,-7.25) -- (N01\i);
        }
        \fill[black!20!lightgray!60] (-9.0-2.1,-8.75)
        rectangle (-9.0+2.1,-10.25);
        \node (N023) at (-9.0,-9.5) {$\VecSp{T}\oplus W_{0,2,3}$};
        \draw[thick] (8.0-17.0,-7.25) -- (N023);
        \fill[black!20!lightgray!60] (-2.0-2.1,-8.75)
        rectangle (-2.0+2.1,-10.25);
        \node (N123) at (-2.0,-9.5) {$\VecSp{T}\oplus W_{1,2,3}$};
        \draw[thick] (-2.0,-7.25) -- (N123);
    \end{tikzpicture}
\end{center}
	\caption{Tree of recursive calls in an exhaustive search with depth $r-\ell=3$}
    \label{Graph:karan1}
\end{figure}
\section{Improving on~\myalgorithm{BDEZ} using symmetries}
\label{sec:bdezstab}
We present in this section, with kind permission from the authors,
an unpublished improvement~\cite{KaraNWithOrbits} to Algorithm~\myalgorithm{BDEZ}.
This improvement takes into account the fact that we can define rank-preserving automorphisms of
$\Bl{m}{n}$.
Their action is defined in Section~\ref{subsec:isom}.
\subsection{Action of automorphisms on $\Bl{m}{n}$}
\label{subsec:isom}
We work with subspaces of $\Bl{m}{n}$ rather than with bilinear maps, as in Section~\ref{sec:delta}.
We describe in this section the rank-preserving group of automorphisms $\sigma$ acting on subspaces $\VecSp{T} \subset \Bl{m}{n}$,
also referred to as the $\RP$-automorphisms group.

\begin{definition}
	\label{def:automat}
	An element $\sigma = (\mu,\nu) \in \GLGL{m}{n}$
	acts on $\Bl{m}{n}$ via
	$$\Phi \comp \sigma : (\mathbf{a},\mathbf{b}) \mapsto \Phi(\mu(\mathbf{a}),\nu(\mathbf{b})).$$
	Such an element is called $\RP$-automorphism.
\end{definition}
\begin{proposition}
	The action of $\GLGL{m}{n}$ is a group action and its elements are all invertible.
\end{proposition}
\begin{proof}
	For $\sigma = (\mu,\nu),\sigma' = (\mu',\nu') \in \GLGL{m}{n}$ and $\Phi \in \Bl{m}{n}$,
	we have
	$$\forall \mathbf{a},\mathbf{b},\ 
	((\Phi \comp \sigma)\comp \sigma') (\mathbf{a},\mathbf{b}) =
	(\Phi \comp \sigma)(\mu'(\mathbf{a}),\nu'(\mathbf{b})) =
	\Phi(\mu(\mu'(\mathbf{a})),\nu(\nu'(\mathbf{b}))) =
	(\Phi \comp (\sigma\comp \sigma'))(\mathbf{a},\mathbf{b}).
	$$
	Thus, the action that we defined is indeed a group action.
	Since all the elements of $\GLGL{m}{n}$ are invertible, we have automorphisms.
\end{proof}
\begin{proposition}[$\RP$-automorphisms preserve the rank]
\label{prop:preserverank}
	Let $\sigma \in \GLGL{m}{n}$.
	\begin{itemize}
	\item 	For any $\Phi \in \Bl{m}{n}$, we have $\rk(\Phi\comp \sigma) = \rk(\Phi)$.
	\item	For any subspace $\VecSp{T} \subset \Bl{m}{n}$, we also have
	$\rk(\VecSp{T} \comp \sigma) = \rk(\VecSp{T})$.
	\end{itemize}
\end{proposition}
\begin{proof}
	First, let $\phi \in \Bl{m}{n}$ of rank one.
	There exist $\alpha \in \mathcal{L}(K^m;K)$ and $\beta \in \mathcal{L}(K^n;K)$ such that
	$\phi : (\mathbf{a},\mathbf{b}) \mapsto \alpha(\mathbf{a})\cdot \beta(\mathbf{b})$.
	There exist $\mu \in \GL(K^m)$ and $\nu \in \GL(K^n)$ such that
	$\phi \comp \sigma : (\mathbf{a},\mathbf{b}) \mapsto \alpha(\mu(\mathbf{a}))\cdot \beta(\nu(\mathbf{b}))$.
	Since $\alpha \comp \mu \in \mathcal{L}(K^m;K)$ and $\beta \comp \nu \in \mathcal{L}(K^n;K)$,
	$\phi \comp \sigma$ is a rank-one bilinear form.

	Since the $\RP$-automorphisms in Definition~\ref{def:automat} preserve the rank of rank-one bilinear forms,
	by linearity and by definition of the rank of a bilinear form, it preserves the rank of any bilinear form.
	For any subspace $\VecSp{T} \subset \Bl{m}{n}$ and any $\sigma \in \GLGL{m}{n}$, we have
	$$\rk(\VecSp{T}\comp \sigma) = \rk(\VecSp{T}).$$
\end{proof}
\begin{remark}
	Note that, when $m=n$,
	Proposition~\ref{def:automat} is not the most general notion of $\RP$-automorphisms that we may have:
	for simplicity, we do not take into account the possible transposition $\tau$ acting on any $\Phi \in \Bl{m}{m}$, via
	$\Phi \circ \tau : (\mathbf{a},\mathbf{b}) \mapsto \Phi(\mathbf{b},\mathbf{a})$.
\end{remark}
\begin{notation}[Group action on matrices]
	The group
	$\GLGL{m}{n}$ is isomorphic to the group $\GL_m(K)\times \GL_n(K)$, acting on matrices $M$ via
	$M \cdot (X,Y) = \Trans{X} \cdot M \cdot Y.$
	Thus, we often consider elements of $\GLGL{m}{n}$ as elements of $\GL_m(K) \times \GL_n(K)$ and
	vice versa.
\end{notation}

\begin{example}[Action of $\GLGL{2}{2}$]
	Let us consider the subspace $V$ of $\Bl{2}{2}$ generated by the bilinear forms represented by the
matrices $M_1$ and
$M_2$ defined as
\begin{equation*}
\begin{split}
M_1 = \begin{pmatrix}
1 & 0 \\
0 & 0
\end{pmatrix},
M_2 = \begin{pmatrix}
0 & 1 \\
0 & 0 
\end{pmatrix}.
\end{split}
\end{equation*}
We take $\sigma = (X,Y)$ such that $X = Y = \begin{pmatrix}
0 & 1 \\
1 & 0
\end{pmatrix}$.

The subspace $V' = V \circ \sigma$ is generated by $M'_1$ and $M'_2$, defined as
\begin{equation*}
\begin{split}
	M'_1= M_1 \cdot \sigma = \Trans{X}\cdot M_1 \cdot Y = \begin{pmatrix}
0 & 0 \\
0 & 1 
\end{pmatrix},\ 
	M'_2 = M_2 \cdot \sigma = \Trans{X}\cdot M_2 \cdot Y = \begin{pmatrix}
0 & 0 \\
1 & 0
\end{pmatrix}.
\end{split}
\end{equation*}
\end{example}

Since we will often refer to subgroups of $\GLGL{m}{n}$ stabilizing elements of $\Bl{m}{n}$ in the following, we
define the notion of setwise stabilizer.
\begin{definition}[Setwise stabilizer]
	For a subset $\mathcal{T}\subset \Bl{m}{n}$, we denote by $\stb(\mathcal{T})$ the subgroup of $\GLGL{m}{n}$
	stabilizing $\mathcal{T}$:
	$$\stb(\mathcal{T}) = \left\{\sigma \in \GLGL{m}{n}\ |\ \mathcal{T}\circ \sigma = \mathcal{T}\right\}.$$
	We use the same notation for a subspace $\VecSp{T}\subset \Bl{m}{n}$.
\end{definition}
In the rest of this work, we often refer to the ``stabilizer'' of a given set
$\mathcal{T}$.
Each time, we exclusively mean
the setwise stabilizer of $\mathcal{T}$, which is, in general, different from
the pointwise stabilizer of $\mathcal{T}$. Indeed, the pointwise stabilizer of $\mathcal{T}$ is defined as
$\{\sigma \in \GLGL{m}{n}\ |\ \forall \Phi \in \mathcal{T},\ \Phi \circ \sigma
 = \Phi\}$.

The algorithmic improvement originally presented in~\cite{KaraNWithOrbits} comes from the fact that,
for any target space $\VecSp{T} \subset \Bl{m}{n}$ of dimension $\ell$ and any integer $r \geq \ell$,
we have 
$$\forall \sigma \in \stb(\VecSp{T}),\ \CCover{r}{\VecSp{T}} \circ \sigma = \CCover{r}{\VecSp{T}},$$
because $\sigma$ preserves the rank.
Thus, we can restrict our interest to the computation of the quotient
$\quotient{\CCover{r}{\VecSp{T}}}{\stb(\VecSp{T})}$
instead of $\CCover{r}{\VecSp{T}}$.

\subsection{\myalgorithm{BDEZ} with stabilizer}
In order to find all the elements of $\CCover{r}{\VecSp{T}}$, it is sufficient to
obtain one representative per equivalence class of $\ALQ{\VecSp{T}}{r}$, from which one can recover the whole orbits
through the group action of $\stb(\VecSp{T})$.
Moreover, we can compute $\quotient{\CCover{r}{\VecSp{T}}}{\stb(\VecSp{T})}$ faster than $\CCover{r}{\VecSp{T}}$.
Thus, we adapt our general strategy to this idea.

\paragraph{\textbf{General strategy for computing the bilinear rank using $\RP$-automorphisms}}
The new algorithmic strategy we are
considering is stated as follows,
for a target subspace $\VecSp{T} \subset \Bl{m}{n}$ of dimension $\ell$ and the associated subgroup $\stb(\VecSp{T})$ of $\RP$-automorphisms
stabilizing $\VecSp{T}$:
\begin{itemize}
\item start with an initial guess $r = \ell$;
\item compute $\quotient{\CCover{r}{\VecSp{T}}}{\stb(\VecSp{T})}$ (the set $\CCover{r}{\VecSp{T}}$ up to the action of $\stb(\VecSp{T})$);
\item if $\quotient{\CCover{r}{\VecSp{T}}}{\stb(\VecSp{T})} = \emptyset$, increment $r$ and return to the previous step;
\item enumerate $\CCover{r}{\VecSp{T}}$ using the action of $\stb(\VecSp{T})$;
\item at the end, $r$ is the rank and $\CCover{r}{\VecSp{T}}$ the set of optimal decompositions.
\end{itemize}
\par

Algorithm~\myalgorithm{BDEZStab} is a recursive approach for the computation of one representative per equivalence class.
The input of the first call to~\myalgorithm{BDEZStab} is: a target subspace $\VecSp{T}$ of dimension $\ell$,
the group $\stb(\VecSp{T})$ and
an integer $r\geq \ell$.
\begin{algorithm}
	\caption{\myalgorithm{BDEZStab}}
	\label{alg:recstabapproc}
	\begin{algorithmic}[1]
		\Require{$\VecSp{T}\subset \Bl{m}{n}$, the stabilizer $\stb(\VecSp{T})$, an integer $r$}
		\Ensure{A set of representatives of $\quotient{\CCover{r}{\VecSp{T}}}{\stb(\VecSp{T})}$}
		\Function{$\myalgorithm{ExpandSubspace}$}{$V,\mathcal{H},U,d,r$} \Comment $\VecSp{V} \subset \Bl{m}{n}, \mathcal{H} \subset \mathcal{G},
		U \subset \stb(\VecSp{T}), r \in\mathbb{N}$
	\If{$d = r$ and $\dim \VecSp{V} = r$ and \myalgorithm{HasRankOneBasis}($\VecSp{V}$)}
		\State \Return $\{\VecSp{V}\}$
	\Else
		\State $\mathcal{S} \leftarrow \emptyset$
		\State $\mathcal{O} \leftarrow \quotient{\mathcal{H}}{U}$ 
		\Comment $\phi$ and $\phi'$ lie in the same orbit if $V\oplus \spn(\{\phi\}) = (V \oplus \spn(\{\phi'\})) \circ \sigma$
		\label{alg:lineorbitscomp}
		\For{$i \in \closedinterval{0}{\#\mathcal{O}-1}$}  \Comment $\mathcal{O} = \left\{O_i\  | \ i \in \closedinterval{0}{\#\mathcal{O}-1}\right\}$ 
\label{recstabline7}
		\State $\phi \leftarrow \myalgorithm{Representative}(O_i)$ \Comment Choose a representative of the orbit $O_i$
		\State $U' \leftarrow \stb{(\VecSp{V}\oplus \spn(\{\phi\}))} \cap U$\label{line:computestab}
		\State $\mathcal{H}' \leftarrow \cup_{j \geq i} O_j$
		\State $\mathcal{S} \leftarrow \mathcal{S}\ \cup\ $\Call{\myalgorithm{ExpandSubspace}}{$\VecSp{V},\mathcal{H}',
		U', d+1, r$} 
		\EndFor
		\State\Return $\mathcal{S}$
		\EndIf
		\EndFunction
		\State \Return \Call{\myalgorithm{ExpandSubspace}}{$\VecSp{T},\mathcal{G},\stb(\VecSp{T}),\ell, r$}
\end{algorithmic}
\end{algorithm}

Figure~\ref{Graph:karan2} describes this recursive approach using a tree
and illustrates how some branches are pruned, relying on Proposition~\ref{prop:stabaction}.
We assume that the initial set of rank-one bilinear forms is $\{\phi_0,\phi_1,\phi_2,\phi_3\}$ and that we have $\sigma \in \stb(\VecSp{T})$ such that
$\sigma(\phi_0) = \phi_1$, $\sigma(\phi_1) = \phi_0$, $\sigma(\phi_2) = \phi_3$ and $\sigma(\phi_3) = \phi_2$.

\begin{proposition}
	\label{prop:stabaction}
	Let $\VecSp{T}$ and $\VecSp{V}$ be subspaces of $\Bl{m}{n}$ such that $\VecSp{V} \in \CCover{r}{\VecSp{T}}$.
	Then, given the orbit $\phi \circ \stb(\VecSp{T})$ of a bilinear form $\phi$ of rank one,
	if $\VecSp{V}$ satisfies
	$\VecSp{V}\cap \left(\phi \circ \stb(\VecSp{T})\right)\neq \emptyset$,
	then there exists an element $\VecSp{V}'$ in the equivalence class of $\VecSp{V}$
	for the action of $\stb(\VecSp{T})$ and such that $\phi \in \VecSp{V}'$.
\end{proposition}
\begin{proof}
	There exists $\sigma \in \stb(\VecSp{T})$ such that $\phi\circ \sigma \in \VecSp{V}$.
	We can then take $\VecSp{V}' = \VecSp{V} \circ (\sigma^{-1})$, which meets all the conditions.
\end{proof}

The particularity of \myalgorithm{BDEZStab} is that, instead of enumerating all the elements of $\mathcal{H}$ as
in \myalgorithm{BDEZ}, we restrict the enumeration to one element per equivalence class for the action of
$U \subset \stb(V)$.
We use in particular the fact that the additional computations such as stabilizers on Line~\ref{line:computestab}
are negligible, compared to the speed-up obtained by pruning branches in~\myalgorithm{BDEZ}.
Heuristically, \myalgorithm{BDEZStab} is faster than \myalgorithm{BDEZ} by a factor $\#\stb{(\VecSp{T})}$.
This method constitutes the state of the art for the current work:
our contribution is compared to the performance of this algorithm.

\begin{figure}[H]
\begin{center}
    \begin{tikzpicture}[scale=0.40,font=\small]

	    \fill[black!20!lightgray!60] (-0.5,0)
        rectangle (0.5,1);
        \node at (0,0.5) {$\VecSp{T}$};
        \foreach\i in {0,2} {
    	    \fill[black!20!lightgray!60] (4*\i-6.0-1.50,-2.75)
	        rectangle (4*\i-6.0+1.50,-4.25);
		\node (N\i) at (4*\i-6.0,-3.5) {$\VecSp{T}\oplus W_{\i}$};
	        \draw[thick] (0,0) -- (4*\i-6.0,-2.75);
        }
        \foreach\i in {1,3} {
    	    \fill[black!20!lightgray!20] (4*\i-6.0-1.50,-2.75)
	        rectangle (4*\i-6.0+1.50,-4.25);
		\node (N\i) at (4*\i-6.0,-3.5) {$\VecSp{T}\oplus W_{\i}$};
	        \draw[gray,dashed] (0,0) -- (4*\i-6.0,-2.75);
        }
        \foreach\i in {1,2,3} {
       	    \fill[black!20!lightgray!60] (4*\i-18.0-1.75,-5.75)
	        rectangle (4*\i-18.0+1.75,-7.25);
	        \node (N0\i) at (4*\i-18.0,-6.5) {$\VecSp{T}\oplus W_{0,\i}$};
            \draw[thick] (-6.0,-4.25) -- (4*\i-17.0,-5.75);
        }
        \foreach\i in {2,3} {
       	    \fill[black!20!lightgray!20] (7.0+4*\i-17.0-1.75,-5.75)
	        rectangle (7.0+4*\i-17.0+1.75,-7.25);
	        \node (N1\i) at (7.0+4*\i-17.0,-6.5) {$\VecSp{T}\oplus W_{1,\i}$};
            \draw[gray,dashed] (4-6.0,-4.25) -- (N1\i);
        }
        \fill[black!20!lightgray!60] (6.0-1.75,-5.75)
        rectangle (6.0+1.75,-7.25);
        \node (N23) at (6.0,-6.5) {$\VecSp{T}\oplus W_{2,3}$};
        \draw[thick] (2.0,-4.25) -- (N23);
        \foreach\i in {2} {
        	\fill[black!20!lightgray!60] (4*\i*1.1-19.0-8.0-2.1,-8.75)
	        rectangle (4*\i*1.1-19.0-8.0+2.1,-10.25);
	        \node (N01\i) at (4*\i*1.1-19.0-8.0,-9.5) {$\VecSp{T}\oplus W_{0,1,\i}$};
            \draw[thick] (4-18,-7.25) -- (N01\i);
        }
        \foreach\i in {3} {
        	\fill[black!20!lightgray!20] (4*\i*1.1-19.0-8.0-2.1,-8.75)
	        rectangle (4*\i*1.1-19.0-8.0+2.1,-10.25);
	        \node (N01\i) at (4*\i*1.1-19.0-8.0,-9.5) {$\VecSp{T}\oplus W_{0,1,\i}$};
            \draw[gray,dashed] (4-18,-7.25) -- (N01\i);
        }
        \fill[black!20!lightgray!60] (-9.0-2.1,-8.75)
        rectangle (-9.0+2.1,-10.25);
        \node (N023) at (-9.0,-9.5) {$\VecSp{T}\oplus W_{0,2,3}$};
        \draw[thick] (8.0-17.0,-7.25) -- (N023);
        \fill[black!20!lightgray!20] (-2.0-2.1,-8.75)
        rectangle (-2.0+2.1,-10.25);
        \node (N123) at (-2.0,-9.5) {$\VecSp{T}\oplus W_{1,2,3}$};
        \draw[gray,dashed] (-2.0,-7.25) -- (N123);
	\draw[thick,blue,->] (-5.75,-4.25) to[bend right]node[pos=0.5,black,scale=1.5,yshift=-1.00mm]{\scriptsize{$\sigma$}} (-2.25,-4.25);
	\draw[thick,blue,->] (3.0,-4.25) to[bend right]node[pos=0.5,black,scale=1.5,yshift=-1.00mm]{\scriptsize{$\sigma$}} (6.50,-4.25);
	\draw[thick,blue,->] (9-26.0,-10.25) to[bend right]node[pos=0.5,black,scale=1.5,yshift=-2.00mm]{\scriptsize{$\sigma$}} (N013);
    \end{tikzpicture}
\end{center}
	\caption{Pruning branches in an exhaustive search using $\RP$-automorphisms.}
    \label{Graph:karan2}
\end{figure}
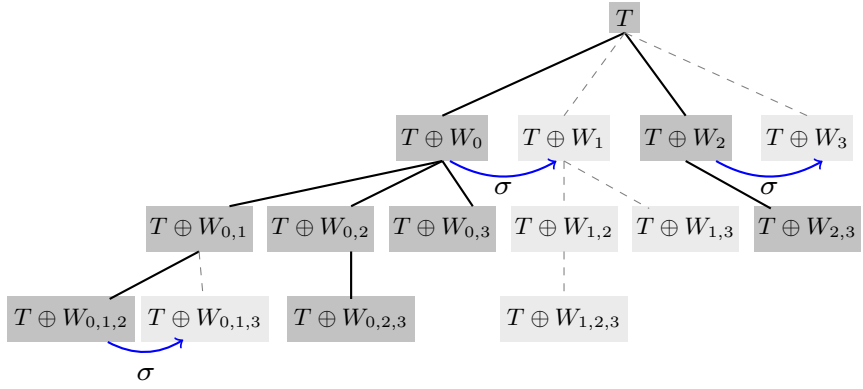
\section{Algebraic structure of some bilinear maps}
\label{sec:algstruct}
In this section, we describe the structure of some vector spaces corresponding
to bilinear maps that are considered in our applications.
This section can be skipped on a first read. It is needed to prove properties of
specific bilinear maps that are stated in Section~\ref{sec:newalgo}.
In particular, we need to know the structure of the stabilizer of a vector space in order
to be able to improve on the exhaustive search.
\subsection{Short product}
\label{subsec:structstabshort}
For the purpose of this section, we restrict our discussion to the specific case defined
as follows. Let $\ell$ be a positive integer, let $\boldsymbol{\Phi}$
be the bilinear map $\boldsymbol{\Phi}\in \mathcal{L}(K^\ell,K^\ell;K^\ell)$
defined by the short product 
$$\boldsymbol{\Phi} : \begin{pmatrix}a_0\\\vdots\\a_{\ell-1}\end{pmatrix},\begin{pmatrix}b_0\\\vdots\\b_{\ell-1}\end{pmatrix}
\mapsto \begin{pmatrix}c_0\\\vdots\\c_{\ell-1}\end{pmatrix}$$
such that $\sum_{0\leq i<\ell} c_iX^{\ell-1-i} = (\sum_{0\leq i<\ell} a_iX^i)(\sum_{0\leq i<\ell} b_iX^{\ell-1-i}) \mod X^\ell$.
Let $\VecSp{T}$ be the subspace of $\Bl{\ell}{\ell}$ spanned by the $\ell$ bilinear forms that
are the coordinates of $\boldsymbol{\Phi}$, denoted by
$\Phi_0,\ldots,\Phi_{\ell-1}$.

The matrix representing the element $\sum_{0\leq i < \ell} m_i \Phi_i \in \VecSp{T}$,
where $m_i \in K$, is
\begin{center}
$M(m_0,\ldots,m_{\ell-1}) =$
\begin{tikzpicture}[baseline={([yshift=-1ex]current bounding box.center)}]
\matrix (m) [inner sep=1.5pt, column sep=6pt, row sep=6pt,matrix of math nodes,nodes in empty cells,right delimiter={]},left delimiter={[} ]{
m_{0} &  m_{1} &    & m_{\ell-1}       \\
0      &             &    &           \\
       &             &    & m_{1}\\
0      &             & 0  & m_{0}    \\
} ;
\draw[loosely dotted] (m-1-4) -- (m-1-2);
\draw[loosely dotted] (m-1-4) -- (m-3-4);
\draw[loosely dotted] (m-3-4) -- (m-1-2);
\draw[loosely dotted] (m-4-4) -- (m-1-1);
\draw[loosely dotted] (m-4-3) -- (m-2-1);
\draw[loosely dotted] (m-4-3) -- (m-4-1);
\draw[loosely dotted] (m-2-1) -- (m-4-1);
\end{tikzpicture},
\end{center}
in the canonical basis.
This matrix is an upper triangular Toeplitz matrix.

Let $N$ be the matrix $M(0,1,\ldots,0)$.
The matrix $N$ is a nilpotent matrix such that
$$\forall j \in \closedinterval{0}{\ell-1},\ M(0,\ldots,0,1,\underbrace{0,\ldots,0}_{j\ \text{zeros}}) = N^{\ell-1-j},$$
and $N^\ell = 0$.
The elements of the algebra $K[N]$ are the upper triangular Toeplitz matrices and $K[N] \cong K[X]/(X^\ell)$.

We provide in Theorem~\ref{thm:stabshortprod} a useful property describing the action of $\stb(\ShProdMap{\ell})$
on $\ShProdMap{\ell}$.

\begin{restatable}{theorem}{stabshortprod}
\label{thm:stabshortprod}
	Let any integer $\ell\geq 2$:
	\begin{enumerate}
		\item the orbit of the identity matrix $I=N^0$ for the action of $\stb(\VecSp{T})$ is the set
			of invertible matrices of $\VecSp{T}$;
		\item the orbit of $N$ for the action of $\stb(\VecSp{T})\cap \stb(I)$ is the set
			of nilpotent matrices of $\VecSp{T}$;
		\item for any pair $(\Psi,\Psi')$ of elements of $\ShProdMap{\ell}$ such that $\rk(\Psi) = \ell$
	and $\rk(\Psi') = \ell-1$, there exists $\sigma \in \stb(\ShProdMap{\ell})$ such that
	$$( \Psi \circ \sigma, \Psi' \circ \sigma) = (I,N);$$
\item we have $\stb(I)\cap \stb(N) \subset \stb(\VecSp{T})$ and the cardinality of $\stb(\VecSp{T})$ is $(\#K)^{3\ell-4} (\#K-1)^3$.
	\end{enumerate}
\end{restatable}
\begin{proof}
	See~\ref{app:stabshort}.
\end{proof}

\subsection{Matrix product}
\label{sec:strucmatprod}
We denote by $\boldsymbol{\Phi}_{p,q,r}$ the bilinear map corresponding to the $p\times q$ by $q\times r$ matrix product:
$$\begin{matrix}
	\boldsymbol{\Phi}_{p,q,r}:& \mathcal{M}_{p,q}(K)\times \mathcal{M}_{q,r}(K) &\longrightarrow &\mathcal{M}_{p,r}(K)\\
	&(A,B) &\longmapsto& A\cdot B
\end{matrix}.$$
We denote by $\Phi_{i,j}$ the bilinear forms such that
$\Phi_{i,j}(A,B)$ is the coefficient $(i,j)$ of $\boldsymbol{\Phi}_{p,q,r}(A,B)$ for $i\in\closedinterval{0}{p-1},j \in \closedinterval{0}{r-1}$.
The elements $\Phi_{i,j}$ satisfy
$\Phi_{i,j}(A,B) = \sum_{0\leq h < q}a_{i,h}b_{h,j}$.

The bilinear map $\boldsymbol{\Phi}_{p,q,r}$ is represented by a subspace of $\Bl{pq}{qr}$
denoted by
$$\VecSp{T}_{p,q,r} = \spn(\{\Phi_{i,j}\}_{\substack{i,j}}).$$

In order to represent the elements of $\VecSp{T}_{p,q,r}$ in terms of matrices
of $\mathcal{M}_{pq,qr}$, we need an order on the $a_{i,h}$'s and $b_{h,j}$'s.
\begin{itemize}
	\item For the $a_{i,h}$'s, we fix the following order:
	$a_{i,h} \leq a_{i',h'}$ if $i\leq i'$ or $i=i'$ and $h\leq h'$, which is the row-major order.
	\item For the $b_{h,j}$'s, we fix the following order:
	$b_{h,j} \leq b_{h',j'}$ if $j\leq j'$ or $j=j'$ and $h\leq h'$, which is the column-major order.
\end{itemize}
Then, in the bases of $\mathcal{M}_{p,q}$ and $\mathcal{M}_{q,r}$ given by the $a_{i,h}$'s and $b_{h,j}$'s
ordered as above, the elements of $\VecSp{T}_{p,q,r}$ can be represented as matrices of $\mathcal{M}_{pq,qr}$
divided in blocks of size $q\times q$ equal to $I_q$ the identity matrix of $\mathcal{M}_{q,q}$.
Consequently, this space is isomorphic to $\mathcal{M}_{p,r} \otimes I_q$ and all the elements of $\VecSp{T}_{p,q,r}$
have a rank which is multiple of $q$.

\begin{example}[Matrix representation of elements of $\VecSp{T}_{2,2,2}$]
The elements of $\VecSp{T}_{2,2,2}$ are represented by matrices of $\mathcal{M}_{4,4}$
spanned by
$$
\begin{pmatrix}
1 & 0 & 0 & 0\\
0 & 1 & 0 & 0\\
0 & 0 & 0 & 0\\
0 & 0 & 0 & 0\\
\end{pmatrix},
\begin{pmatrix}
0 & 0 & 1 & 0\\
0 & 0 & 0 & 1\\
0 & 0 & 0 & 0\\
0 & 0 & 0 & 0\\
\end{pmatrix},
\begin{pmatrix}
0 & 0 & 0 & 0\\
0 & 0 & 0 & 0\\
1 & 0 & 0 & 0\\
0 & 1 & 0 & 0\\
\end{pmatrix},
\begin{pmatrix}
0 & 0 & 0 & 0\\
0 & 0 & 0 & 0\\
0 & 0 & 1 & 0\\
0 & 0 & 0 & 1\\
\end{pmatrix},
$$
corresponding to the coefficients $a_{0,0}b_{0,0}+a_{0,1}b_{1,0}$, $a_{0,0}b_{0,1}+a_{0,1}b_{1,1}$,
$a_{1,0}b_{0,0}+a_{1,1}b_{1,0}$ and $a_{1,0}b_{0,1}+a_{1,1}b_{1,1}$, respectively.
The previous matrices can also be expressed as
$$
\begin{pmatrix}
1 & 0\\
0 & 0
\end{pmatrix} \otimes I_2,
\begin{pmatrix}
0 & 1\\
0 & 0
\end{pmatrix} \otimes I_2,
\begin{pmatrix}
0 & 0\\
1 & 0
\end{pmatrix} \otimes I_2,
\begin{pmatrix}
0 & 0\\
0 & 1
\end{pmatrix} \otimes I_2,
$$
respectively.
\end{example}

Let $(e_i)$, $(f_h)$ and $(g_j)$ be the canonical bases of $K^p$, $K^q$ and $K^r$.
The subspace $\VecSp{T}_{p,q,r}$ can be easily characterized with the tensor notation:
it is generated by the vectors, for $i\in\closedinterval{0}{p-1},j \in \closedinterval{0}{r-1}$,
$$\Phi_{i,j} = \sum_{0\leq h< q}e_i\otimes f_h \otimes f_h \otimes g_j.$$

\begin{restatable}{theorem}{stabmat}
	\label{thm:stabmat}
	For the group action $M \cdot (X,Y) \mapsto \Trans{X} M Y$, the subgroup stabilizing
	the vector space $\VecSp{T}_{p,q,r}$ can be described as the group given by the pairs
	$({P\otimes \Trans{R}},Q\otimes (R^{-1}))$ for $P \in \GL_p$, $R\in \GL_q$, and $Q\in \GL_r$.
\end{restatable}
\begin{proof}
	See~\ref{app:stabmat}.
\end{proof}
\begin{corollary}
	\label{prop:stabmat}
The elements of $\VecSp{T}_{p,q,r}$ of a given rank lie in the same orbit under
the action of $\stb(\VecSp{T}_{p,q,r})$.
\end{corollary}
\begin{example}[Action of the stabilizer of $\VecSp{T}_{2,2,2}$]
	The stabilizer of $\VecSp{T}_{2,2,2}$ is generated by the following elements of $\GLGL{4}{4}$:
	$$\left(
	\begin{pmatrix}
		1 & 0 & 1 & 0\\
		0 & 1 & 0 & 1\\
		0 & 0 & 1 & 0\\
		0 & 0 & 0 & 1\\
	\end{pmatrix},
	\begin{pmatrix}
		1 & 0 & 0 & 0\\
		0 & 1 & 0 & 0\\
		0 & 0 & 1 & 0\\
		0 & 0 & 0 & 1\\
	\end{pmatrix}
	\right),\ 
	\left(
	\begin{pmatrix}
		0 & 0 & 1 & 0\\
		0 & 0 & 0 & 1\\
		1 & 0 & 0 & 0\\
		0 & 1 & 0 & 0\\
	\end{pmatrix},
	\begin{pmatrix}
		1 & 0 & 0 & 0\\
		0 & 1 & 0 & 0\\
		0 & 0 & 1 & 0\\
		0 & 0 & 0 & 1\\
	\end{pmatrix}
	\right),$$
	$$\left(
	\begin{pmatrix}
		1 & 1 & 0 & 0\\
		0 & 1 & 0 & 0\\
		0 & 0 & 1 & 1\\
		0 & 0 & 0 & 1\\
	\end{pmatrix},
	\begin{pmatrix}
		1 & 1 & 0 & 0\\
		0 & 1 & 0 & 0\\
		0 & 0 & 1 & 1\\
		0 & 0 & 0 & 1\\
	\end{pmatrix}
	\right),\ 
	\left(
	\begin{pmatrix}
		0 & 1 & 0 & 0\\
		1 & 0 & 0 & 0\\
		0 & 0 & 0 & 1\\
		0 & 0 & 1 & 0\\
	\end{pmatrix},
	\begin{pmatrix}
		0 & 1 & 0 & 0\\
		1 & 0 & 0 & 0\\
		0 & 0 & 0 & 1\\
		0 & 0 & 1 & 0\\
	\end{pmatrix}
	\right),$$
	$$\left(
	\begin{pmatrix}
		1 & 0 & 0 & 0\\
		0 & 1 & 0 & 0\\
		0 & 0 & 1 & 0\\
		0 & 0 & 0 & 1\\
	\end{pmatrix},
	\begin{pmatrix}
		1 & 0 & 1 & 0\\
		0 & 1 & 0 & 1\\
		0 & 0 & 1 & 0\\
		0 & 0 & 0 & 1\\
	\end{pmatrix}
	\right),\ 
	\left(
	\begin{pmatrix}
		1 & 0 & 0 & 0\\
		0 & 1 & 0 & 0\\
		0 & 0 & 1 & 0\\
		0 & 0 & 0 & 1\\
	\end{pmatrix},
	\begin{pmatrix}
		0 & 0 & 1 & 0\\
		0 & 0 & 0 & 1\\
		1 & 0 & 0 & 0\\
		0 & 1 & 0 & 0\\
	\end{pmatrix}
	\right)
	.
	$$
The vector space of $\VecSp{T}_{2,2,2}$ is isomorphic to $\mathcal{M}_{2,2} \otimes I_2$.
Thus the elements of $\VecSp{T}_{2,2,2}$ have rank $0$, $2$ or $4$.

Via the action of $\stb(\VecSp{T}_{2,2,2})$, all the elements of rank $2$ can all be mapped to the element
$$
\begin{pmatrix}
1 & 0 & 0 & 0\\
0 & 1 & 0 & 0\\
0 & 0 & 0 & 0\\
0 & 0 & 0 & 0\\
\end{pmatrix}.
$$
Similarly, via the action of $\stb(\VecSp{T}_{2,2,2})$, all the elements of rank $4$ can all be mapped to the element
$$
\begin{pmatrix}
1 & 0 & 0 & 0\\
0 & 1 & 0 & 0\\
0 & 0 & 1 & 0\\
0 & 0 & 0 & 1\\
\end{pmatrix}.
$$
\end{example}
\section{Coverings of subspaces of bilinear forms}
\label{sec:newalgo}
Our contribution consists in reducing the number of vector spaces $\VecSp{W}$ that we need to
enumerate in order to get those that satisfy $\VecSp{T}\oplus\VecSp{W}\in\AL{r}$, where $\VecSp{T}$
is the vector space representing a given bilinear map.
To this effect, we restrict the enumeration to vector spaces $\VecSp{W}$ satisfying some properties which are
intrinsic to $\VecSp{T}$.
In this section, the definition and theoretical aspects of the set of vector spaces satisfying these
properties are treated, illustrated via the example of the short product and the matrix product.
In Section~\ref{sec:computetopo},
we deal with practical and computational aspects.

\subsection{Theoretical aspect}
\label{sec:covsets}
Our strategy consists, first, for any $r\geq \ell$, in constructing $g$ sets $\mathcal{E}_{i,r}$ for $i \in \closedinterval{0}{g-1}$, that are all subsets
of $\AL{r-\ell+k_i}$, where $k_i$ is a nonnegative integer, and that
satisfy some property described in Definition~\ref{def:coveringdef}.
\begin{definition}[Covering of a vector space]
	\label{def:coveringdef}
	Let $r$ be a nonnegative integer, and $\{k_i\}$ a set of nonnegative integers
	such that $k_i \leq \ell$.
	Let $\VecSp{T}$ be a subspace of $\Bl{m}{n}$ of dimension $\ell$.
	Let $\{\mathcal{E}_{i,r}\}_{0\leq i < g}$ be
	a set of subsets where
	$\mathcal{E}_{i,r} \subset \AL{r-\ell+k_i}$, for all $i \in
	\closedinterval{0}{g-1}$. Then, $(\mathcal{E}_{i,r})_{0\leq i < g}$ 
	is said to be a covering of $\VecSp{T}$ if and only if,
	for any vector space $\VecSp{W} \in \AL{r-\ell}$
	such that $\VecSp{T}\oplus \VecSp{W} \in \AL{r}$, there exist an index
	$i \in \closedinterval{0}{g-1}$, a subspace $\VecSp{V} \in \mathcal{E}_{i,r}$, and an $\RP$-automorphism $\sigma \in \stb(\VecSp{T})$ such that $\VecSp{T}
	+ (\VecSp{V}\circ \sigma) = \VecSp{T} \oplus \VecSp{W}$.
\end{definition}
\begin{proposition}
	Given $\VecSp{T}\subset \Bl{m}{n}$ as above and a covering $(\mathcal{E}_{i,r})_{0\leq i < g}$ of 
	${\VecSp{T}}$, then, for any $r\geq \ell$, we
	have
	$$\CCover{r}{\VecSp{T}} \subset \{\VecSp{T} + \VecSp{V}\ |\ 
\exists i \in \closedinterval{0}{g-1},\ \VecSp{V} \in \mathcal{E}_{i,r}\circ \stb(\VecSp{T})\}.$$
\end{proposition}
\begin{proof}
	Let $\VecSp{V} \in \CCover{r}{\VecSp{T}}$. By Proposition~\ref{prop:baseincomplete}, there exists $\VecSp{W}\in
	\AL{r-\ell}$ such that $\VecSp{T} \oplus \VecSp{W} = \VecSp{U}$.
	Then, by Definition~\ref{def:coveringdef}, there exist an index
	$i \in \closedinterval{0}{g-1}$, a subspace $\VecSp{V} \in \mathcal{E}_{i,r}$,
	and an $\RP$-automorphism $\sigma \in \stb(\VecSp{T})$ such that $\VecSp{T}
	+ (\VecSp{V}\circ \sigma) = \VecSp{T} \oplus \VecSp{W}$.
	Taking $\VecSp{V}' = \VecSp{V} \circ \sigma$, we thus have
	$\VecSp{U} = \VecSp{T} + \VecSp{V}'$ and $\VecSp{V}' \in \mathcal{E}_{i,r} \circ \stb(\VecSp{T})$, which proves the inclusion.
\end{proof}
Thus, assuming that we have
a method for computing the $\mathcal{E}_{i,r}$'s, we are able to cover the whole set $\CCover{r}{\VecSp{T}}$.
For example, the set composed of the single set
$\mathcal{E}_{0,r} = \quotient{\AL{r-\ell}}{\stb(\VecSp{T})}$ is a covering of
${\VecSp{T}}$ and can be enumerated using \myalgorithm{BDEZStab}.
We describe below
how we construct the $\mathcal{E}_{i,r}$'s that we use in practice.

\begin{definition}[Stem of a vector space]
	\label{def:propertyintr}
	For a vector space $\VecSp{T}$, a set $\{\VecSp{F}_i\}_{0\leq i <g}$ of $g$ subspaces $\VecSp{F}_i \subset \VecSp{T}$
	of dimension $k_i$ is said to be a \emph{stem} of $\VecSp{T}$
	if and only if,
		for any basis $\mathcal{B}$ of $\VecSp{T}$,
			there exist $i \in \closedinterval{0}{g-1}$, an $\RP$-automorphism $\sigma\in\stb(\VecSp{T})$ and 
			a free family $\mathcal{F} \subset \mathcal{B}$ of size $k_i$ such that
		\begin{equation*}
		\label{eq:propertyintr}
			\spn(\mathcal{F})\circ \sigma = \VecSp{F}_i.
		\end{equation*}
\end{definition}
\begin{proposition}
	\label{prop:formeirs}
	For a vector space $\VecSp{T}\subset \Bl{m}{n}$, a stem of $\VecSp{T}$
	given by $g$ subspaces $\VecSp{F}_i \subset \VecSp{T}$, and
	$g$ subgroups $U_i \subset \stb(\VecSp{T})\cap \stb(\VecSp{F}_i)$,
	the set $\{\mathcal{E}_{i,r}\}_{0\leq i <g}$, where each
	$\mathcal{E}_{i,r}$ is a set of representatives of the quotient $\quotient{\CCover{r-\ell+k_i}{\VecSp{F}_i}}{U_i}$,
	is a covering of $\VecSp{T}$.
\end{proposition}
\begin{proof}
Let $\VecSp{W}\in \AL{r-\ell}$ be such that
$\VecSp{T}\oplus\VecSp{W} \in \AL{r}$.
Take a basis $\mathcal{W}$ of $\VecSp{W}$, and complete it into a basis
of $\VecSp{T}\oplus \VecSp{W}$ using $\ell$ rank-one bilinear forms, denoted
by $\{\psi_i\}_{0\leq i<\ell}$. For all $i \in \closedinterval{0}{\ell-1}$,
write $\psi_i = t_i + w_i$, with $t_i\in \VecSp{T}$ and $w_i \in \VecSp{W}$.

The $t_i$'s are linearly independent. Otherwise, there would exist
coefficients $(\lambda_i)_{0\leq i< \ell}$ such that
$\sum_{i=0}^{\ell-1} \lambda_i t_i = 0$, whence
$\sum_{i=0}^{\ell-1} \lambda_i \psi_i =\sum_{i=0}^{\ell-1} \lambda_i w_i$,
which would then contradict the fact that $\{\psi_i\}_{0\leq i< \ell}$
completes $\mathcal{W}$ into a basis of $\VecSp{T}\oplus \VecSp{W}$.

Consequently, $\mathcal{B} = \{t_i\}_{0\leq i <\ell}$ is a free
family of $\ell$ vectors of $\VecSp{T}$ and, as $\dim(\VecSp{T}) = \ell$,
$\mathcal{B}$ is a basis of $\VecSp{T}$.
Then, by Definition~\ref{def:propertyintr}, there exist an index $i
\in \closedinterval{0}{g-1}$, a subset $\mathcal{F}\subset \mathcal{B}$
of size $k_i = \dim(\VecSp{F}_i)$, and an $\RP$-automorphism $\sigma \in
\stb(\VecSp{T})$ such that $\spn(\mathcal{F}) \circ \sigma = \VecSp{F}_i$.

Let $\VecSp{V} = \VecSp{W} \oplus \spn(\mathcal{F})$. Writing $\mathcal{F} =
\{t_i\}_{i \in I}$, with $I \subset \closedinterval{0}{\ell-1}$,
we define $\mathcal{F}' = \{\psi_i\}_{i\in I}$.
Since $\psi_i = t_i +w_i$ and $\spn(\mathcal{F}') \in \AL{k_i}$,
we have $\VecSp{V} = \VecSp{W}\oplus \spn(\mathcal{F}) = \VecSp{W}
\oplus\spn(\mathcal{F}') \in \AL{r-\ell+k_i}$.

Now, consider $\VecSp{V}' = \VecSp{V} \circ \sigma = (\VecSp{W} \oplus
\spn(\mathcal{F}))\circ \sigma$: we also have $\VecSp{V}' \in \AL{r-\ell+k_i}$,
as $\RP$-automorphisms preserve the bilinear rank, and $\VecSp{F}_i = \spn(\mathcal{F}) \circ \sigma \subset \VecSp{V}'$, whence $\VecSp{V}' \in \CCover{r-\ell+k_i}{\VecSp{F}_i}$.

Finally, let $\VecSp{V}'' \in \mathcal{E}_{i,r}$ be a representative of
the equivalence class of $\VecSp{V}'$ in the quotient set
$\quotient{\CCover{r-\ell+k_i}{\VecSp{F}_i}}{U_i}$:
there exists an $\RP$-automorphism $\gamma \in U_i$ such that $\VecSp{V}'' = \VecSp{V}'\circ \gamma$. We then have

$$\VecSp{T} + (\VecSp{V}'' \circ \gamma^{-1}\circ \sigma^{-1})
 = \VecSp{T} + (\VecSp{V}'\circ \sigma^{-1}) = \VecSp{T} + \VecSp{V}=
\VecSp{T} + (\VecSp{W}\oplus \spn(\mathcal{F})) = \VecSp{T} \oplus \VecSp{W}
$$
where the last equality comes from the fact that
$\spn(\mathcal{F}) \subset \VecSp{T}$.
Finally, as $\gamma^{-1}\circ \sigma^{-1} \in \stb(\VecSp{T})$,
this proves the result.
\end{proof}

Given $\VecSp{T}$ and a stem of $\VecSp{T}$, we can derive a new algorithm that computes $\CCover{r}{\VecSp{T}}$
via the computation of some intermediate sets
$\mathcal{E}_{i,r} = \quotient{\CCover{r-\ell+k_i}{\VecSp{F}_i}}{U_i}$ for $i\in \closedinterval{0}{g-1}$.
\begin{example}[Two examples of stems]
	For any vector space $\VecSp{T}$,
	let $\mathcal{B}$ be a basis of $\VecSp{T}$. There exists a subset of $\mathcal{B}$
	generating $\VecSp{T}$ (namely, $\mathcal{B}$): $\{\VecSp{T}\}$ is a stem of $\VecSp{T}$.
	There exists also a subset of $\mathcal{B}$ generating $\spn(\emptyset)$
	(namely, $\emptyset$):
	$\{\spn(\emptyset)\}$ is a stem of $\VecSp{T}$.
	\begin{itemize}
		\item An enumeration algorithm that uses $\{\VecSp{T}\}$ as a stem
		amounts to computing $\quotient{\CCover{r}{\VecSp{T}}}{\stb(\VecSp{T})}$.
		In this case, we did not decompose the original problem into simpler problems.
		\item If the stem chosen is the set $\{\spn(\emptyset)\}$, this is equivalent to enumerate a set
		of representatives of the quotient
		$\quotient{\CCover{r-\ell}{\spn(\emptyset)}}{\stb(\VecSp{T})}$.
		For this purpose, no better methods than \emph{\myalgorithm{BDEZStab}} is known.
	\end{itemize}
\end{example}
Thus, \myalgorithm{BDEZStab} can be seen as an approach derived from the stem $\{\spn(\emptyset)\}$.
We propose here other strategies that are derived from stems, given by sets of subspaces $\VecSp{F}_i \subset \VecSp{T}$
of dimension $k_i$. The enumeration of a set $\CCover{r-\ell+k_i}{\VecSp{F}_i}$ is interesting in practice if
its cardinality is less than $\#\AL{r-\ell}$. However,
its cost depends also on the algorithms used for the computation of quotients and stabilizers and
on how large $k_i$ is, which is detailed below.

No automatic method is known to determine, how to choose a stem for a given vector space $\VecSp{T}$:
we have to provide a stem for each $\VecSp{T}$. This task has to be done by hand specifically for each
bilinear map. We will actually do so in Section~\ref{subsec:ideaalgo} and~\ref{sec:apptoprods} for the examples of the short product and the matrix.
To this end, the determination of the stabilizer, as done in Section~\ref{sec:algstruct}, plays a key role.

In order to compute a set of the form $\quotient{\CCover{r-\ell+k_i}{\VecSp{F}_i}}{U_i}$,
we proceed in two steps. Let $\mathcal{F}_i$ be a basis of $\VecSp{F}_i$. Our strategy assumes that
we have a finite representation of a group $U_i$ such that $U_i\subset \stb(\VecSp{T})\cap \stb(\VecSp{F}_i)$.
In Proposition~\ref{prop:formeirs}, the larger the groups $U_i$ are, the smaller the $\mathcal{E}_{i,r}$'s are.
And we prefer to keep the $\mathcal{E}_{i,r}$'s as small as possible, since it gives smaller sets to enumerate.
Thus, this should lead us to choose
$U_i = \stb(\VecSp{T}) \cap \stb(\VecSp{F}_i)$.
However, in practice, the method used in our implementation is specialized to the choice
$U_i = \stb(\VecSp{T}) \cap \stb(\mathcal{F}_i) \subset \stb(\VecSp{T}) \cap \stb(\VecSp{F}_i)$
(we have $\stb(\mathcal{F}_i) \subset\stb(\VecSp{F}_i)$) because only in this case do
we have a practical algorithm to enumerate a set of representatives for the quotient
$\quotient{\CCover{r-\ell+k_i}{\VecSp{F}_i}}{U_i}$.

\begin{notation}
	For a free family $\mathcal{F}$ of $k$ bilinear forms and a positive integer $d$, 
	we let
	$$\QQCover{d+k}{\mathcal{F}} = \quotient{\Cover{d+k}{\mathcal{F}}}{\stb(\mathcal{F})}.$$
\end{notation}
In order to enumerate sets of the form $\quotient{\CCover{r-\ell+k_i}{\VecSp{F}_i}}{\stb(\VecSp{T})\cap
\stb(\mathcal{F}_i)}$, we adopt a two-step strategy.
\begin{remark}
This strategy requires the precomputation of a set of representatives of the quotient
$$\quotient{\AL{r-\ell+k_i}}{\GLGL{m}{n}}.$$
Section~\ref{sec:combobj} describes how to compute such a set.

However, there is a pratical limit on their dimension $k_i$, due to the precomputations that are used
in our method and that constitute a bottleneck.
Assuming that
$$\#\left(\quotient{\AL{d}}{\GLGL{d}{d}}\right)$$
behaves as $(d!)^{1.1}$ over $\mathbb{F}_2$ (which is an empirical estimate),
storing a set of representatives of
$$\quotient{\AL{d}}{\GLGL{d}{d}}$$
for $d=13$ would require $15$ terabytes for instance.
Consequently, given the largest
``$d$'' for which we are able to compute in practice
$$\quotient{\AL{d}}{\GLGL{m}{n}},$$
we have a practical constraint on how large the $r-\ell+k_i$'s may be:
we should have $r-\ell + k_i \leq d$ for all $i$.\smallskip
\end{remark}
Thus, we precompute the quotient $\quotient{\AL{r-\ell+k_i}}{\GLGL{m}{n}}$.
The first step consists in computing $\QQCover{r-\ell+k_i}{\mathcal{F}_i}$ and is detailed in Section~\ref{sec:computetopoone}.
The second step applies the action of the left transversal
$$\quotient{\stb(\mathcal{F}_i)}{\stb(\VecSp{T})\cap \stb(\mathcal{F}_i)},$$
which can be computed using the algorithms proposed in~\cite{opac-b1102239} for example.

We describe in Algorithm~\myalgorithm{CoveringSetsMethod} the global strategy to find optimal formulae for $\VecSp{T}$ 
in the sense of the bilinear rank,
that is, to enumerate $\CCover{r}{\VecSp{T}}$ given a stem.
We assume that we are given a subspace $\VecSp{T}$ and a set of $g$ free families $\mathcal{F}_0,\ldots,\mathcal{F}_{g-1}$ of $\VecSp{T}$
such that $\{\spn(\mathcal{F}_i)\}_i$ forms a stem of $\VecSp{T}$.
\begin{algorithm}[H]
	\caption{\myalgorithm{CoveringSetsMethod}}
	\label{alg:mainalg}
	\begin{algorithmic}[1]
		\Require{$\VecSp{T}\in\Bl{m}{n}$ of dimension $\ell$, an integer $r$, a stem $\{\mathcal{F}_i\}_{0\leq i< g}$, the sets $\left\{\quotient{\AL{r-\ell+k_i}}{\GLGL{m}{n}}\right\}_{0 \leq i < g}$}
		\Ensure{$\CCover{r}{\VecSp{T}}$}
		\State $\mathcal{S} \leftarrow \emptyset$
		\For{$i \in \closedinterval{0}{g-1}$}
			\State $\mathcal{Q} \leftarrow \QQCover{r-\ell+k_i}{\mathcal{F}_i}$, obtained from
			$\quotient{\AL{r-\ell+k_i}}{\GLGL{m}{n}}$
			\label{line:computeq}
			\Comment See Section~\ref{sec:computetopo}
			\State $\mathcal{L}\leftarrow\quotient{\stb(\mathcal{F}_i)}{\stb(\VecSp{T})\cap\stb(\mathcal{F}_i)}$
			\label{line:computer}
			\For{$\VecSp{W} \in \mathcal{Q}, \sigma \in \mathcal{L}$}
			\If{\myalgorithm{HasRankOneBasis}($\VecSp{T} + (\VecSp{W}\circ\sigma)$)}
				\State $\mathcal{S} \leftarrow \mathcal{S} \cup \{\VecSp{T}+(\VecSp{W}\circ \sigma)\}$
			\EndIf
			\EndFor
		\EndFor
		\State \Return $\bigcup_{\VecSp{V} \in \mathcal{S}}\VecSp{V} \circ \stb(\VecSp{T})$
	\end{algorithmic}
\end{algorithm}

\begin{theorem}
	\label{thm:mainthm}
	Let $R$ be the rank of $\VecSp{T}$.
	For any positive integer $r$,  Algorithm~\myalgorithm{CoveringSetsMethod} proves either that
	\begin{enumerate}
		\item $r <R$
		\item or $R \leq r$.
	\end{enumerate}
	In the case where $R \leq r$,
	any element of $\CCover{r}{\VecSp{T}}$ is included in the set returned by Algorithm~\myalgorithm{CoveringSetsMethod}.
\end{theorem}

The computation of the quotient $\mathcal{Q}$ on Line~\ref{line:computeq} is detailed in Section~\ref{sec:computetopoone}.
\subsection{A stem for the short product}
\label{subsec:ideaalgo}
We use the same notations as 
in Section~\ref{subsec:structstabshort}: we denote by $\Phi_0,\ldots,\Phi_{\ell-1}$ the bilinear forms such that
$$\forall i \geq 0,\ \Phi_i(A,B) = \sum_{j\in \closedinterval{0}{\ell-1-i}} a_{j}b_{j+i}$$
and by $\ShProdMap{\ell}$ the subspace $\spn(\{\Phi_0,\ldots,\Phi_{\ell-1}\})$.

In order to produce a covering of the vector spaces $\VecSp{W}$ satisfying
$\ShProdMap{\ell}\oplus \VecSp{W} \in \CCover{r}{\VecSp{T}}$
that we compute with~\myalgorithm{CoveringSetsMethod}, we need a stem of $\VecSp{T}$.
This stem is given in Proposition~\ref{prop:covshortprod}.
\begin{proposition}[Stem for the short product]
	\label{prop:covshortprod}
	For any $\ell\geq 2$ the singleton $\{\spn(\{\Phi_{0},\Phi_{1}\})\}$is a stem of ${\ShProdMap{\ell}}$:
	for any basis $\mathcal{B}$ of $\ShProdMap{\ell}$,
	there exists $\sigma \in \stb(\ShProdMap{\ell})$ and $\mathcal{F}\subset \mathcal{B}$ of cardinality $2$
	such that
	$$\spn(\mathcal{F}) \circ \sigma = \spn(\{\Phi_{0},\Phi_{1}\}).$$
\end{proposition}
\begin{proof}
	We first observe that for any $\Phi \in \spn(\{\Phi_{\ell-1-i},\ldots,\Phi_{\ell-1}\})$, $\rk(\Phi) \leq i+1$.
	Therefore, any element of rank $\ell$ in $\ShProdMap{\ell}$ has a nonzero coordinate over $\Phi_{0}$
	in its decomposition over the basis $(\Phi_0,\ldots,\Phi_{\ell-1})$ and, reciprocally,
	any element having a nonzero coordinate over $\Phi_{0}$ has rank $\ell$.
	Thus, a basis $\mathcal{B}$ of $\VecSp{T}$ necessarily contains an element of rank $\ell$ denoted by $\Psi$.
	The element $\Psi$ has a nonzero coordinate over $\Phi_{0}$, when we decompose it over $\left\{\Phi_0,\ldots,
	\Phi_{\ell-1}\right\}$.
	Similarly, there exist $\Psi' \in \mathcal{B}$ and $\lambda \in K$ for which
	$\Psi'-\lambda \Psi$ has  rank $\ell-1$.

	We then use Theorem~\ref{thm:stabshortprod} to find an element $\sigma \in \stb(\ShProdMap{\ell})$ such that
	$$(\Psi \circ \sigma,\Psi' \circ \sigma) = (\Phi_{0},\Phi_{1})\ \text{or}\ 
	(\Psi \circ \sigma,(\Psi-\lambda\Psi') \circ \sigma) = (\Phi_{0},\Phi_{1}),$$
	which concludes.
\end{proof}
We give in Table~\ref{tab:compshortprod3} the cardinality of coverings of $\CCover{r}{\ShProdMap{\ell}}$
given by Proposition~\ref{prop:covshortprod}.
\begin{table}[H]
	\centering
	\begin{tabular}{|c|c|}
		\hline
		\textbf{set} & \textbf{cardinality}\\
		\hline
		\hline
		$\Cover{2}{\emptyset} = \AL{2}$ & $980$ \\
		\hline
		\hline
		$\Cover{3}{\Phi_{0}}$ & $28$\\
		\hline
		\hline
		$\Cover{4}{\Phi_{0},\Phi_{1}}$ & $6$\\
		\hline
	\end{tabular}
	\caption{Comparison of the cardinality for $\ell=3$ of three coverings of ${\ShProdMap{3}}$ for $K=\mathbb{F}_2$.}
	\label{tab:compshortprod3}
\end{table}

In conclusion, we need to compute the following set: $\QCover{r-\ell+2}{\Phi_{0},\Phi_{1}}$.
We describe in Section~\ref{sec:computetopo} how  we perform Line~\ref{line:computeq} of Algorithm~\myalgorithm{CoveringSetsMethod}.
The set $\mathcal{L}$ on Line~\ref{line:computer} of \myalgorithm{CoveringSetsMethod} is, for the short product, a set containing one element, which is the identity element of $\GLGL{\ell}{\ell}$.
\subsection{A stem for the matrix product $3\times 2$ by $2\times 3$ over $\mathbb{F}_2$}
\label{sec:apptoprods}
We focus here on the special case given by the bilinear map
$$\begin{matrix}
	\boldsymbol{\Phi}_{3,2,3}:& \mathcal{M}_{3,2}(\mathbb{F}_2)\times \mathcal{M}_{2,3}(\mathbb{F}_2) &\longrightarrow &\mathcal{M}_{3,3}(\mathbb{F}_2)\\
	&(A,B) &\longmapsto& A\cdot B
\end{matrix}$$
over $K = \mathbb{F}_2$.
The rank of this bilinear map is known to be $15$~\cite{doi:10.1137/0120004}. However,
all the optimal formulae are not known.
We denote by $\Phi_{i,j}$ the bilinear forms such that
$\Phi_{i,j}(A,B)$ is the coefficient $(i,j)$ of $\boldsymbol{\Phi}_{3,2,3}(A,B)$ for $i,j \in \left\{0,1,2\right\}$.
The elements $\Phi_{i,j}$ satisfy
$\Phi_{i,j}(A,B) = a_{i,0}b_{0,j} + a_{i,1}b_{1,j}$.

The target subspace of $\Bl{6}{6}$ considered is denoted by
$$\VecSp{T}_{3,2,3} = \spn(\{\Phi_{i,j}\}_{\substack{i,j\in\left\{0,1,2\right\}}}).$$
The approach proposed in this section can be generalized to any matrix product
(albeit at the expense of combinatorial blowup).

We use the stem of $\VecSp{T}_{3,2,3}$  given by Proposition~\ref{prop:covmatprod}.
\begin{proposition}[Covering of the matrix product]
\label{prop:covmatprod}
	The set
	$$\mathcal{C} = \{\spn(\{\Phi_{0,0}+\Phi_{1,1}+\Phi_{2,2}\}),\spn(\{\Phi_{0,0}+\Phi_{1,1},\Phi_{0,1}+\Phi_{2,2}\}),
	\spn(\{\Phi_{0,0}+\Phi_{1,1},\Phi_{1,1}+\Phi_{2,2}\}),$$
	$$\spn(\{\Phi_{0,0}+\Phi_{1,1},\Phi_{2,2}\}),\spn(\{\Phi_{0,0},\Phi_{1,1},\Phi_{2,2}\})\}$$
	is a covering of $\VecSp{T}_{3,2,3}$:
	for any basis $\mathcal{B}$ of $\VecSp{T}_{3,2,3}$, there exists $\mathcal{F} \subset \mathcal{B}$ and
	$\sigma \in \stb(\VecSp{T}_{3,2,3})$ such that
	$$\spn(\mathcal{F})\circ \sigma \in \mathcal{C}.$$
\end{proposition}
\begin{proof}
Let $\mathcal{B}$ be a basis of $\VecSp{T}_{3,2,3}$.
\begin{itemize}
	\item 
		If there exists an element $\Phi$ of rank $6$ in $\mathcal{B}$, then,
		according to Corollary~\ref{prop:stabmat}, there exists $\sigma \in \stb(\VecSp{T}_{3,2,3})$
		such that $\Phi_{0,0}+\Phi_{1,1}+\Phi_{2,2} \in \mathcal{B}\circ \sigma$.
		Otherwise, any element $\Phi$ of $\mathcal{B}$ has rank smaller or equal to $4$ and we have to distinguish two cases.
	\item If there exists an element $\Phi$ of rank $4$, there exists $\sigma$ such that
		$\Phi_{0,0}+\Phi_{1,1}\in \mathcal{B}\circ \sigma$
		and, consequently, there exists another element $\Phi' \in \mathcal{B}$ of rank $2$ or $4$ whose coordinate
		over $\Phi_{2,2}$ in the basis $(\Phi_{i,j})_{i,j}$ is nonzero:
		we need to look at the possible orbits in which $\Phi'$ is included under the action of the subgroup
		of $\stb(\VecSp{T}_{3,2,3})$ preserving the fact that $\Phi$ is in the orbit of $\Phi_{0,0}+\Phi_{1,1}$.
		We can prove that there exist $3$ such orbits and that there exists $\sigma \in \stb(\VecSp{T}_{3,2,3})$
		and $\mathcal{F}\subset \mathcal{B}$ of cardinality $2$
		such that
		$$\mathcal{F}\circ \sigma = \begin{cases}\{\Phi_{0,0}+\Phi_{1,1},\Phi_{0,1}+\Phi_{2,2}\}\\
								\hfil \text{or}\\
                                                           \{\Phi_{0,0}+\Phi_{1,1},\Phi_{1,1}+\Phi_{2,2}\}\\
								\hfil \text{or}\\
							      \{\Phi_{0,0}+\Phi_{1,1},\Phi_{2,2}\}.\\
		\end{cases}$$
	\item Otherwise, all the elements of $\mathcal{B}$ have rank $2$ and there exists $\mathcal{F}\subset \mathcal{B}$ and $\sigma
		\in \stb(\VecSp{T}_{3,2,3})$ such that
		$$\mathcal{F}\circ \sigma = \{\Phi_{0,0},\Phi_{1,1},\Phi_{2,2}\}.$$
\end{itemize}
\end{proof}

\section{How to compute subspaces containing specific bilinear forms}
\label{sec:computetopo}
We propose in this section a method for computing a covering of
${\CCover{r}{\VecSp{T}}}$,
where $\VecSp{T}$ is a target space of dimension $\ell$. The covering is a set of subspaces containing a specific set
of bilinear forms described as in Section~\ref{subsec:ideaalgo} or~\ref{sec:apptoprods}.
More specifically, we are interested in computing sets defined as $\QCover{r-\ell+k}{\Psi_0,\ldots,\Psi_{k-1}}$,
for $\Psi_0,\ldots,\Psi_{k-1}$ bilinear forms of $\Bl{m}{n}$. Those can be described as sets of subspaces of rank $r-\ell+k$
containing a prescribed set $\left\{\Psi_0,\ldots,\Psi_{k-1}\right\}$ of bilinear forms, up to the action of $\stb(\{\Psi_0,\ldots,\Psi_{k-1}\})$.
\subsection{General approach}
\label{sec:computetopoone}
First, our strategy consists in precomputing the quotient
$\quotient{\AL{m,n,r-\ell+k}}{\GLGL{m}{n}}$.
This quotient is smaller than $\AL{m,n,r-\ell+k}$ by construction.
We explain how to compute it in Section~\ref{sec:combobj}.

Algorithm~\ref{alg:computq} explains how we compute the quotient $\mathcal{Q}$ in
Algorithm~\myalgorithm{CoveringSetsMethod}.

\begin{algorithm}
	\caption{\myalgorithm{IntermediateSetViaQuotientComputation}}
	\label{alg:computq}
	\begin{algorithmic}[1]
		\Require{$\quotient{\AL{m,n,r-\ell+k}}{\GLGL{m}{n}},\{\Psi_0,\ldots,\Psi_{k-1}\} = 
		\mathcal{F}$}
		\Ensure{$\mathcal{Q}$ a sest of representatives per orbit of $\QQCover{r-\ell+k}{\mathcal{F}}$}
		\State $\mathcal{Q} \leftarrow \emptyset$
		\For{$\VecSp{W} \in \quotient{\AL{m,n,r-\ell+k}}{\GLGL{m}{n}}$}
		\label{line:typeofquot}
			\For{$\quotient{\left\{\{\Phi_0,\ldots,\Phi_{k-1}\} \subset\VecSp{W}\ |\ \forall t,\ \rk(\Phi_t) = \rk(\Psi_t)\right\}}{\stb(\VecSp{W})}$} \label{line:quotientofkelements}
				\If{$\exists \sigma \in \GLGL{m}{n},\ \{\Phi_0,\ldots,\Phi_{k-1}\} \circ \sigma =\{\Psi_0,\ldots,\Psi_{k-1}\}$} \label{line:findsigma}
					\State $\mathcal{Q} \leftarrow \mathcal{Q}\cup \{\VecSp{W}\circ \sigma\}$
				\EndIf
			\EndFor
		\EndFor
		\State \Return $\mathcal{Q}$
	\end{algorithmic}
\end{algorithm}
\begin{proof}[Correctness of Algorithm~\ref{alg:computq}]
	By construction, according to Line~\ref{line:findsigma}, any element of $\mathcal{Q}$ is an element of
	$$\CCover{r-\ell+k}{\{\Psi_0,\ldots,\Psi_{k-1}\}}.$$

	\begin{itemize}
		\item First, we prove that {any orbit of $\QCover{r-\ell+k}{\Psi_0,\ldots,\Psi_{k-1}}$
			has a representative in $\mathcal{Q}$}.

	Let $\VecSp{W}'$ be a representative of an orbit in 
	$\QCover{r-\ell+k}{\Psi_0,\ldots,\Psi_{k-1}}.$
	There exist $\sigma \in \GLGL{m}{n}$ and
	$\VecSp{W}$ a representative of an element of $\quotient{\AL{m,n,r-\ell+k}}{\GLGL{m}{n}}$ such that
	$\VecSp{W}\circ\sigma = \VecSp{W}'$.
	Thus, we have $\{\Psi_0,\ldots,\Psi_{k-1}\}\circ \sigma^{-1} \subset \VecSp{W}$
	and the set $$\left\{\Psi_0 \circ \sigma^{-1},\ldots,\Psi_{k-1}\circ \sigma^{-1} \right\}$$
	satisfies the predicate	on Line~\ref{line:findsigma}.
	Any $\sigma'$ such that $ \{\Psi_0,\ldots,\Psi_{k-1}\}\circ \sigma^{-1}\circ \sigma' = \{\Psi_0,\ldots,\Psi_{k-1}\}$ satisfies
	$$\sigma' \in \sigma\circ \stb(\{\Psi_0,\ldots,\Psi_{k-1}\}),$$
	which means that an element of $\VecSp{W}\circ \sigma\circ\stb(\{\Psi_0,\ldots,\Psi_{k-1}\}) = \VecSp{W}' \circ \stb(\{\Psi_0,\ldots,\Psi_{k-1}\})$
	is included in the list returned by Algorithm~\ref{alg:computq}.
	Thus, the list returned contains at least one representative per orbit of $\QQCover{r-\ell+k}{\{\Psi_0,\ldots,\Psi_{k-1}\}}$.
		\item In the following, we prove that each orbit of $\QCover{r-\ell+k}{\Psi_0,\ldots,\Psi_{k-1}}$ has a {unique}
	representative in $\mathcal{Q}$.

	Assume that there exist $\VecSp{W},\VecSp{W}'\in \mathcal{Q}$ and $\gamma \in \stb(\{\Psi_0,\ldots,\Psi_{k-1}\})$ such that $\VecSp{W} = \VecSp{W}' \circ \gamma$.
By construction, there exists $\VecSp{W}_0,\VecSp{W}'_0 \in \AL{r-\ell+k}$ and $\sigma,\sigma' \in \GLGL{m}{n}$
	such that $\VecSp{W} = \VecSp{W}_0 \circ \sigma$ and $\VecSp{W}' = \VecSp{W}'_0
	\circ \sigma'$. Then $\VecSp{W}'_0 = \VecSp{W}_0 \circ \sigma \circ \gamma^{-1} \circ
	\sigma'^{-1}$, whence $\VecSp{W}'_0 = \VecSp{W}_0$ as
	on Line~\ref{line:typeofquot} of Algorithm~\ref{alg:computq} we enumerate only one representative
	of each orbit of $\quotient{\AL{r-\ell+k}}{\GLGL{m}{n}}$.
	Thus, $\sigma \circ \gamma^{-1} \circ 
	\sigma'^{-1} \in \stb(\VecSp{W}_0)$.

	Still by construction, there exists $\{\Phi_0,\ldots,\Phi_{k-1}\}$ and 
	$\{\Phi'_0,\ldots,\Phi'_{k-1}\} \subset \VecSp{W}_0$ such that
	$$\{\Phi_0,\ldots,\Phi_{k-1}\} \circ \sigma = \{\Psi_0,\ldots,\Psi_{k-1}\}$$
	and
	$$\{\Phi'_0,\ldots,\Phi'_{k-1}\} \circ \sigma' = \{\Psi_0,\ldots,\Psi_{k-1}\}.$$
	Then,
	\begin{equation*}
	\begin{split}
	\{\Phi'_0,\ldots,\Phi'_{k-1}\} &= \{\Psi_0,\ldots,\Psi_{k-1}\} \circ \sigma'^{-1}
	= \{\Psi_0,\ldots,\Psi_{k-1}\} \circ \gamma^{-1} \circ \sigma'^{-1}
	= \{\Phi_0,\ldots,\Phi_{k-1}\} \circ \sigma \circ\gamma^{-1} \circ \sigma'^{-1}
	\end{split}
	\end{equation*}
	and $\{\Phi_0,\ldots,\Phi_{k-1}\}$ is in the same orbit as
	$\{\Phi'_0,\ldots,\Phi'_{k-1}\}$ under the action of $\stb(\VecSp{W}_0)$,
	which is contradictory with the definition of the quotient
	on Line~\ref{line:quotientofkelements}.
	\end{itemize}
\end{proof}
Testing the predicate on Line~\ref{line:findsigma} is a problem generalizing the problem of~\cite[Ch.\ 19]{Brgisser:2010:ACT:1965416}
and~\cite{doi:10.1137/0208037}:
given two pairs $(M_0,M_1)$ and $(N_0,N_1)$ of $(\mathcal{M}_{m,n})^2$, determine whether
there exists two invertible matrices $X$ and $Y$ such that $(\Trans{X} M_0 Y,\Trans{X} M_1 Y)= (N_0,N_1)$,
which is done by computing a Weierstrass--Kronecker canonical form for $(M_0,M_1)$.
When we consider more than two matrices, for example three matrices $(M_0,M_1,M_2)$ mapped to $(N_0,N_1,N_2)$,
we compute $(X,Y)$ such that $(M_0,M_1)$ is mapped to $(N_0,N_1)$ and we compose it with elements of
$\quotient{\stb{M_0}\cap \stb{M_1}}{\stb{M_2}}$, computed with the algorithms proposed in~\cite{opac-b1102239} for example.
The complexity for finding all the $\RP$-automorphisms $\sigma$ in \myalgorithm{IntermediateSetViaQuotientComputation}
is bounded by the cardinality of $\AL{r-\ell+k}$ (which is comparable
to~\myalgorithm{BDEZ}) by construction,
and is hard to estimate more precisely. In our applications, it appears to be negligible
compared to~\myalgorithm{BDEZ}.

\subsection{Application to the short product}
We come back to the example given in Section~\ref{subsec:ideaalgo} corresponding to the short product.
We recall that $\VecSp{T}$ is the subspace obtained from the bilinear map
given by the short product modulo $\ell$ and that we need to compute the set
$\mathcal{Q}=\QCover{r-\ell+2}{\Phi_{0},\Phi_{1}}$
for a given integer $r$.

If we take $\ell=3$, we can represent $\Phi_0$ and $\Phi_1$ by the matrices
$$I = \begin{pmatrix} 1 & 0 & 0\\ 0& 1 & 0 \\ 0 & 0 & 1\end{pmatrix}\ \text{and}\ N = \begin{pmatrix} 0 & 1 & 0\\ 0& 0 & 1 \\ 0 & 0 & 0\end{pmatrix}.$$
Thus, for a given couple $(M_0,M_1)$ of matrices representing bilinear forms of a subspace
$\VecSp{W} \in\quotient{\AL{r-\ell+2}}{\GLGL{\ell}{\ell}}$ , we are looking for invertible matrices $X$ and $Y$ such that
$$\Trans{X} M_0 Y = I\ \text{and}\ \Trans{X}M_1 Y = N,$$
which is done in Algorithm~\ref{alg:computqshort}. As it is precised on Line~\ref{line:gaussredshort} of Algorithm~\ref{alg:computqshort}, we find $X$ and $Y$ such that
$\Trans{X} M_0 Y = I$ via Gauss reduction. Then, we need to check whether $\Trans{X}M_1 Y$ and $N$ are similar or not
($(\Trans{X}M_1 Y)^\ell$ should be the null matrix for this purpose), as done on Line~\ref{line:similarshort} of Algorithm~\ref{alg:computqshort}.

\begin{algorithm}
	\caption{\myalgorithm{IntermediateSetViaQuotientComputation} (Short product)}
	\label{alg:computqshort}
	\begin{algorithmic}[1]
		\Require{$\quotient{\AL{r-\ell+2}}{\GLGL{\ell}{\ell}}$}
		\Ensure{One representative per orbit of $\mathcal{Q}$, defined as above}
		\State $\mathcal{Q} \leftarrow \emptyset$
		\For{$\VecSp{W} \in \quotient{\AL{\ell,\ell,r-\ell+2}}{\GLGL{\ell}{\ell}}$}
			\For{$\Psi \in \quotient{\{\Phi \in \VecSp{W}\ |\ \rk(\Phi) = \ell\}}{\stb(\VecSp{W})}$}
				\State Let $\sigma$ such that $\Psi \circ \sigma = I$ \Comment We obtain $\sigma$ via a Gauss reduction
				\label{line:gaussredshort}
				\State $\VecSp{W}' \leftarrow \VecSp{W} \circ \sigma$
				\For{$\Psi' \in \quotient{\{ \Phi \in \VecSp{W}'\ |\ \rk(\Phi) = \ell-1\}}{\stb(\VecSp{W}')\cap \stb(I)}$}
					\If{$\exists \sigma' \in \stb(I),\ \Psi'\circ \sigma' = N$}
					\Comment Using that $N$ and $\Psi'$ are similar \label{line:similarshort}
						\State $\mathcal{Q} \leftarrow \mathcal{Q}\cup \{\VecSp{W}\circ \sigma\circ \sigma'\}$
					\EndIf
				\EndFor
			\EndFor
		\EndFor
		\State \Return $\mathcal{Q}$
	\end{algorithmic}
\end{algorithm}

Once we have computed $\mathcal{Q}$, it remains to compute the left transversal
$$\mathcal{L} = \quotient{\stb(\{I,N\})}{\stb(\VecSp{T})\cap\stb(\{I,N\})}$$
and to compute $\mathcal{Q} \circ \mathcal{L}$.
According to Theorem~\ref{thm:stabshortprod}, we have $\#\mathcal{L} = 1$, which means that
Algorithm~\ref{alg:computqshort} actually returns $\QCover{r-\ell+2}{I,N}\circ \mathcal{L}$.

In terms of complexity, we do not have explicit bounds. However, we can state that the complexity depends linearly
on $\#\quotient{\AL{r-\ell+2}}{\GLGL{\ell}{\ell}}$ and on the number of pairs of bilinear forms $(\Phi,\Psi)$ per element of
$\quotient{\AL{r-\ell+2}}{\GLGL{\ell}{\ell}}$ such that $\rk(\Phi) = \ell$ and $\rk(\Psi) = \ell-1$.
\subsection{Computing the orbits of vector spaces of bilinear forms}
\label{sec:combobj}
In this section, we propose an approach for computing the set
$\quotient{\AL{m,n,d}}{\GLGL{m}{n}}$,
required by the algorithm described in Section~\ref{sec:computetopoone}.
Its cost is at least exponential in $d$, $m$ and $n$ and difficult to estimate.

\begin{notation}
	\label{ref:omega}
	We denote by ${\Omega}_d$ the quotient $\quotient{\AL{d,d,d}}{\GLGL{d}{d}}$
	for any $d\geq 1$.
\end{notation}

First, we describe how we represent elements of $\AL{m,n,d}$ and we prove that given the knowledge of $\Omega_d$
we can deduce the elements of
$\quotient{\AL{m,n,d}}{\GLGL{m}{n}}$
for any $m$ and $n$ from this precomputation.

Let $\VecSp{W}$ be an element of $\AL{m,n,d}$. There exist $d$ rank-one bilinear forms $\phi_t : (\mathbf{a},\mathbf{b}) \mapsto
\alpha_t(\mathbf{a})\cdot \beta_t(\mathbf{b})$
such that $\VecSp{W} = \spn \left(\{\phi_i\}_{i \in \closedinterval{0}{d-1}}\right)$.
In the canonical basis of $K^m$ and $K^n$, we represent $\alpha_t$ and $\beta_t$ as matrices
of $\mathcal{M}_{1,m}$ and $\mathcal{M}_{1,n}$.
Thus, there exist two matrices $U\in \mathcal{M}_{d,m}$ and $V\in \mathcal{M}_{d,n}$, whose rows are given by the linear forms
$\alpha_t$ and $\beta_t$ respectively, and $\VecSp{W}$ can be represented by the pair $(U,V)$. Such a
representation is not unique (for example, any permutation of the rows of $(U,V)$ gives a valid representation).
In particular, for a pair of matrices $(U,V)$ representing some vector space $\VecSp{W}$,
there exists $\sigma = \mu\times \nu$ in $\GL(K^m) \times \GL(K^n)$
such that the pair of matrices $U',V'$, such that $(U',V') = (U\circ \mu,V\circ \nu)$
represents $ \VecSp{W} \circ \sigma$, are the reduced column echelon form of the matrices $U$ and $V$, respectively.

\begin{example}
	Let us consider the vector space $\VecSp{W}$ of $\AL{3,4,6}$ generated by the
	rank-one bilinear forms represented by
	\begin{gather*}
	\begin{split}
	M_1 = \begin{pmatrix}
	1 & 1 & 0 &0 \\
	0 & 0 & 0 &0 \\
	0 & 0 & 0 &0 \\
	\end{pmatrix},
	M_2 = \begin{pmatrix}
	0 & 0 & 0 &0 \\
	0 & 0 & 0 &0 \\
	1 & 0 & 0 &0 \\
	\end{pmatrix},
	M_3 = \begin{pmatrix}
	0 & 0 & 0 &0 \\
	1 & 0 & 0 &0 \\
	0 & 0 & 0 &0 \\
	\end{pmatrix},\\
	M_4 = \begin{pmatrix}
	0 & 0 & 0 &0 \\
	0 & 1 & 0 &0 \\
	0 & 0 & 0 &0 \\
	\end{pmatrix},
	M_5 = \begin{pmatrix}
	0 & 0 & 0 &0 \\
	0 & 0 & 1 &0 \\
	0 & 0 & 0 &0 \\
	\end{pmatrix},
	M_6 = \begin{pmatrix}
	0 & 0 & 0 &0 \\
	0 & 0 & 0 &1 \\
	0 & 0 & 0 &0 \\
	\end{pmatrix}.
	\end{split}
	\end{gather*}
	The pair of matrices $(U,V)$ associated to $\VecSp{W}$ is
	$$
	\begin{pmatrix}
		1 & 0 & 0 \\
		0 & 0 & 1 \\
		0 & 1 & 0 \\
		0 & 1 & 0 \\
		0 & 1 & 0 \\
		0 & 1 & 0 \\
	\end{pmatrix}
	,
	\begin{pmatrix}
		1 & 1 & 0 & 0\\
		1 & 0 & 0 & 0\\
		1 & 0 & 0 & 0\\
		0 & 1 & 0 & 0\\
		0 & 0 & 1 & 0\\
		0 & 0 & 0 & 1\\
	\end{pmatrix}
	.$$
\end{example}


%

Assuming that we have a representation of the elements of
$\quotient{\AL{d,d,d}}{\GLGL{d}{d}}$ in terms of pairs of matrices $(U,V)\in \mathcal{M}_{d,d}\times \mathcal{M}_{d,d}$
in reduced column echelon form,
we obtain all the elements of
$$\quotient{\AL{m,n,d}}{\GLGL{m}{n}}$$
by considering the subset $\Omega'_d$ of $\Omega_d = \quotient{\AL{d,d,d}}{\GLGL{d}{d}}$ of elements represented by
matrices $(U,V)$ in reduced column echelon form such that $\rk(U) \leq  \min(m,d)$ and $\rk(V) \leq \min(n,d)$.
Given $m$ and $n$,
a set of representatives for
$$\quotient{\AL{m,n,d}}{\GLGL{m}{n}}$$
can be seen as matrices $(U',V')\in \mathcal{M}_{d,m}\times \mathcal{M}_{d,n}$ in reduced column echelon form
and for which there exists matrices $(U,V) \in \mathcal{M}_{d,d}^2$, representing an element of
$\Omega'_d$, obtained by adding $d-m$ an $d-n$
zero columns to $U'$ and $V'$, respectively, or by removing zero columns if $d<m$ or $d<n$.\\\smallskip

Our strategy consists in deducing $\Omega_d$ from the computation of $\Omega_{d-1}$.
Algorithm~\ref{alg:compomega} describes this strategy: for each
vector space $\VecSp{W}$ of $\Omega_{d-1}$, we extend it to a vector space of
$\Bl{d}{d}$ by padding with zeros, and we consider the vector spaces $\VecSp{W}\oplus \spn(\{\phi\})$
that can be obtained by adding an element $\phi$ of rank one.
We remove from the set of $\VecSp{W}\oplus \spn(\{\phi\})$ the vector spaces that are isomorphic via
an isomorphism test. We determine whether two vector spaces $\VecSp{W}'$ and $\VecSp{W}$
are isomorphic if there exists a basis of $\VecSp{W}'$ of rank-one bilinear forms
such that the corresponding couple of matrices $(U',V')$ in reduced column echelon form is equal to $(U,V)$.
The complexity of this approach depends on the number of bases of rank-one bilinear forms of $\VecSp{W}$,
which, compared to $d$, is not large generically.
However, there are degenerate cases for which the nomber of bases is very large (exponential in $d^2$).
These cases require specific code to recognize them and to treat them separately.

\begin{algorithm}[H]
		\caption{\myalgorithm{IterativeQuotientsComputation}}
	\label{alg:compomega}
	\begin{algorithmic}[1]
	\Require{$\Omega_{d-1}$, a set $\mathcal{G}$ of rank-one bilinear forms}
	\Ensure{$\Omega_d$}
	\State $\widehat{\Omega}_{d-1} \leftarrow \myalgorithm{Extend}(\Omega_{d-1})$
	\Comment We compute extensions of elements of $\Omega_{d-1}$ in $\Bl{d}{d}$
	\State $\mathcal{L} \leftarrow \emptyset$
	\For{$\VecSp{W} \in \widehat{\Omega}_{d-1}$}
		\State $\mathcal{H} \leftarrow \quotient{\mathcal{G}}{\stb \VecSp{W}}$
		\For{$h \in \mathcal{H}$}
			\State $\mathcal{L} \leftarrow \mathcal{L} \cup \{\VecSp{W} \oplus \spn(\{h\})\}$
		\EndFor
	\EndFor
	\State \label{iterquotcompt:return} \Return $\quotient{\mathcal{L}}{\GLGL{d}{d}}$
		\Comment We remove isomorphic vector spaces in $\mathcal{L}$ 
\end{algorithmic}
\end{algorithm}


The naive algorithm which checks for each pair of elements
of the set $\mathcal{L}$ whether or not they are isomorphic, computed in Line~\ref{iterquotcompt:return} of
Algorithm~\ref{alg:compomega},
can be improved. Indeed, we propose to compute invariants for the group action induced by $\GLGL{d}{d}$
and to compare subspaces having the same invariants.
For example,
for $\VecSp{W} \in \AL{d,d,d}$,
we consider the polynomial
$P_{\VecSp{W}} = \sum_{0\leq t \leq d} p_t(\VecSp{W}) X^t$
such that
$$\forall t\geq 0,\ p_t(\VecSp{W}) = \#\{\phi \in \VecSp{W} |\ \rk (\phi) = t \}.$$
Therefore,
for any $\sigma \in \GLGL{d}{d}$, $P_{\VecSp{W} \circ \sigma} = P_{\VecSp{W}}$.

We have been able to compute $\Omega_d$ for $d \in \closedinterval{1}{8}$ and $K= \mathbb{F}_2$
with an implementation in Magma~V2.21-3~\cite{MR1484478}\footnote{\label{first}The code of this implementation
can be found at the address \href{http://karancode.gforge.inria.fr}{http://karancode.gforge.inria.fr}}.
The timings are described in Table~\ref{tab:timeomegad}.

\begin{table}[H]
	\begin{center}
	\begin{tabular}{|c|c|c|c|c|c|c|c|c|}
	\hline
	set & $\Omega_1$ & $\Omega_2$ & $\Omega_3$ & $\Omega_4$ & $\Omega_5$ & $\Omega_6$ & $\Omega_7$ & $\Omega_8$ \\
	\hline
	cardinality & $1$ & $3$ & $9$ & $31$ & $141$ & $969$ & $11,289$ & $265,577$ \\
	\hline
	upper bound & $1$ & $9$ & $4.4 \cdot 10^2$ & $9.9 \cdot 10^{4}$ & $9.5 \cdot 10^{7}$ & $3.8 \cdot 10^{11}$ & $6.1 \cdot 10^{15}$ & $4.0 \cdot 10^{20}$ \\
	\hline
	time (s) & $0$ & $4.0\cdot 10^{-2}$ & $6.0\cdot 10^{-2}$ & $1.8\cdot 10^{-1}$ & $1.5$ & $1.8\cdot 10$ &
	$4.7 \cdot 10^2$ & $1.8\cdot 10^4$ \\
	\hline
	\end{tabular}
	\end{center}
	\caption{Timings for our approach to compute the sets $\Omega_d$ over $K = \mathbb{F}_2$ on
	a single core of a 3.3 GHz Intel Core i5-4590.}
	\label{tab:timeomegad}
\end{table}

It would be interesting to obtain an upper bound on $\#\Omega_d$ with the good order of
magnitude. Indeed, we are able to say for instance that $\#\Omega_d$ is bounded by the quantity
$$\#\Omega_d \leq \left(\#K^d-1\right)^2\cdot \#\Omega_{d-1},$$
corresponding to the number of possible rank-one bilinear forms that we add to elements of $\Omega_{d-1}$
to obtain an element of $\Omega_d$. This formula leads recursively to the following bound:
$$\#\Omega_d \leq \left(\prod_{t\in \closedinterval{1}{d}}(\#K^t-1)\right)^2.$$
However, this upper bound differs by a huge factor from the true cardinality of $\Omega_d$ and
cannot consequently be used in a complexity analysis.
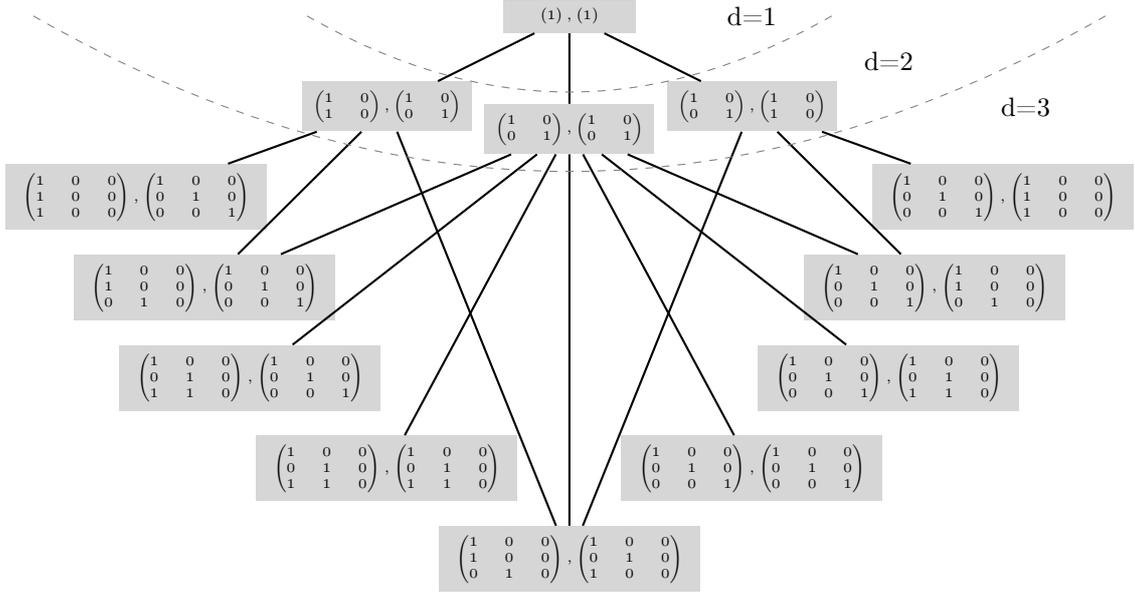
\begin{figure}
	\centering
\begin{tikzpicture}[scale=0.60, font=\tiny]
\node[text width=1.5cm,text centered,rectangle,fill=black!20!lightgray!40](M1) at (0.0,0.0)
{$\begin{pmatrix}
1
\end{pmatrix},
\begin{pmatrix}
1
\end{pmatrix}
$
};
\node[text width=1.5cm,text centered](D1) at (4.0,0.0)
	{\normalsize{d=1}};
\node[text width=1.5cm,text centered](D2) at (7.0,-1.0)
	{\normalsize{d=2}};
\node[text width=1.5cm,text centered](D3) at (10.0,-2.0)
	{\normalsize{d=3}};
\node[text width=2cm,text centered,text centered,rectangle,fill=black!20!lightgray!40] (M21) at (-4.0,-2.0){
$\begin{pmatrix}
	1 & 0\\
	1 & 0\\
\end{pmatrix},
\begin{pmatrix}
	1 & 0\\
	0 & 1\\
\end{pmatrix}
$
};
\node[text width=2cm,text centered,rectangle,fill=black!20!lightgray!40] (M22) at (0.0,-2.5){
$\begin{pmatrix}
	1 & 0\\
	0 & 1\\
\end{pmatrix},
\begin{pmatrix}
	1 & 0\\
	0 & 1\\
\end{pmatrix}
$
};
\node[text width=2cm,text centered,rectangle,fill=black!20!lightgray!40] (M23) at (4.0,-2.0){
$\begin{pmatrix}
	1 & 0\\
	0 & 1\\
\end{pmatrix},
\begin{pmatrix}
	1 & 0\\
	1 & 0\\
\end{pmatrix}
$
};
\node[text width=3.2cm,text centered,rectangle,fill=black!20!lightgray!40] (M31) at (-9.5,-4.0){
$\begin{pmatrix}
	1 & 0 & 0\\
	1 & 0 & 0\\
	1 & 0 & 0\\
\end{pmatrix},
\begin{pmatrix}
	1 & 0 & 0\\
	0 & 1 & 0\\
	0 & 0 & 1\\
\end{pmatrix}
$
};
\node[text width=3.2cm,text centered,rectangle,fill=black!20!lightgray!40] (M32) at (4.0,-10.0){
$\begin{pmatrix}
	1 & 0 & 0\\
	0 & 1 & 0\\
	0 & 0 & 1\\
\end{pmatrix},
\begin{pmatrix}
	1 & 0 & 0\\
	0 & 1 & 0\\
	0 & 0 & 1\\
\end{pmatrix}
$
};

\node[text width=3.2cm,text centered,rectangle,fill=black!20!lightgray!40] (M33) at (9.5,-4.0){
$\begin{pmatrix}
	1 & 0 & 0\\
	0 & 1 & 0\\
	0 & 0 & 1\\
\end{pmatrix},
\begin{pmatrix}
	1 & 0 & 0\\
	1 & 0 & 0\\
	1 & 0 & 0\\
\end{pmatrix}
$
};
\node[text width=3.2cm,text centered,rectangle,fill=black!20!lightgray!40] (M34) at (-8.0,-6.0){
$\begin{pmatrix}
	1 & 0 & 0\\
	1 & 0 & 0\\
	0 & 1 & 0\\
\end{pmatrix},
\begin{pmatrix}
	1 & 0 & 0\\
	0 & 1 & 0\\
	0 & 0 & 1\\
\end{pmatrix}
$
};
\node[text width=3.2cm,text centered,rectangle,fill=black!20!lightgray!40] (M35) at (0.0,-12.0){
$\begin{pmatrix}
	1 & 0 & 0\\
	1 & 0 & 0\\
	0 & 1 & 0\\
\end{pmatrix},
\begin{pmatrix}
	1 & 0 & 0\\
	0 & 1 & 0\\
	1 & 0 & 0\\
\end{pmatrix}
$
};
\node[text width=3.2cm,text centered,rectangle,fill=black!20!lightgray!40] (M36) at (8.0,-6.0){
$\begin{pmatrix}
	1 & 0 & 0\\
	0 & 1 & 0\\
	0 & 0 & 1\\
\end{pmatrix},
\begin{pmatrix}
	1 & 0 & 0\\
	1 & 0 & 0\\
	0 & 1 & 0\\
\end{pmatrix}
$
};
\node[text width=3.2cm,text centered,rectangle,fill=black!20!lightgray!40] (M37) at (-7.0,-8.0){
$\begin{pmatrix}
	1 & 0 & 0\\
	0 & 1 & 0\\
	1 & 1 & 0\\
\end{pmatrix},
\begin{pmatrix}
	1 & 0 & 0\\
	0 & 1 & 0\\
	0 & 0 & 1\\
\end{pmatrix}
$
};
\node[text width=3.2cm,text centered,rectangle,fill=black!20!lightgray!40] (M38) at (-4.0,-10.0){
$\begin{pmatrix}
	1 & 0 & 0\\
	0 & 1 & 0\\
	1 & 1 & 0\\
\end{pmatrix},
\begin{pmatrix}
	1 & 0 & 0\\
	0 & 1 & 0\\
	1 & 1 & 0\\
\end{pmatrix}
$
};
\node[text width=3.2cm,text centered,rectangle,fill=black!20!lightgray!40] (M39) at (7.0,-8.0){
$\begin{pmatrix}
	1 & 0 & 0\\
	0 & 1 & 0\\
	0 & 0 & 1\\
\end{pmatrix},
\begin{pmatrix}
	1 & 0 & 0\\
	0 & 1 & 0\\
	1 & 1 & 0\\
\end{pmatrix}
$
};
\draw[thick] (M1) -- (M21);
\draw[thick] (M1) -- (M22);
\draw[thick] (M1) -- (M23);
\draw[thick] (M21) -- (M31);
\draw[thick] (M21) -- (M34);
\draw[thick] (M21) -- (M35);
\draw[thick] (M23) -- (M33);
\draw[thick] (M23) -- (M36);
\draw[thick] (M23) -- (M35);
\draw[thick] (M22) -- (M32);
\draw[thick] (M22) -- (M37);
\draw[thick] (M22) -- (M38);
\draw[thick] (M22) -- (M39);
\draw[thick] (M22) -- (M34);
\draw[thick] (M22) -- (M36);
\draw[thick] (M22) -- (M35);
	\draw[dashed,gray] (-5.75,0) to[bend right] (5.75,0);
	\draw[dashed,gray] (-11.75,0) to[bend right] (11.75,0);
\end{tikzpicture}
	\caption{Partially ordered structure of the $\Omega_d$ for $d\leq 3$ and $K= \mathbb{F}_2$.}
	\label{fig:poset}
\end{figure}

To conclude, we show in Figure~\ref{fig:poset} how the subspaces of $\Omega_3$ over $\mathbb{F}_2$
are related to $\Omega_2$ and $\Omega_1$
by using its partially ordered set
structure. Each element of $\Omega_d$ is represented by the corresponding  couple of
matrices $(U,V)$ of $\mathcal{M}_{d,d}^2$.

\section{Experimental results}
\label{sec:expresults}
An implementation in Magma~V2.21-3~\cite{MR1484478} of the algorithms presented in
the previous sections has been done\footref{first}. We compare in this section the timings
obtained from various instances of the bilinear rank problem for these different algorithms.
Our Magma implementation of the algorithm described in~\cite{Barbulescu2012} is clearly slower than the original C version.
However, since we are interested in the speed-up obtained from our work, we need 
a fair approach. We show in particular that Algorithm~\myalgorithm{BDEZStab},
although it is neither multithreaded nor written in C, improves
considerably on the timings estimated in~\cite{Barbulescu2012}.
The new algorithm proposed in the current article is denoted by \myalgorithm{CoveringSetsMethod}:
compared to Algorithm~\myalgorithm{BDEZStab}, it constitutes a huge speed-up on particular
instances of the bilinear rank problem among which the matrix product, discussed in Sections~\ref{exp:matprod323}
and~\ref{exp:matprod}, and the short product,
discussed in Section~\ref{exp:shortprod}.
All the timings presented in this section have been done on a single core of a 3.3 GHz Intel Core i5-4590 processor.
\subsection{Recursive approach}
We need a few notations to denote the various bilinear maps we are interested in:
\begin{itemize}
	\item $\myalgorithm{MatProduct}_{(p,q,r)}$ denotes the product of matrices $p\times q$ by $q\times r$,
	\item $\myalgorithm{ShortProduct}_{\ell}$ denotes the product of polynomials modulo $X^\ell$,
	\item $\myalgorithm{CirculantProduct}_{\ell}$ denotes the product of polynomials modulo $X^{\ell}-1$.
\end{itemize}

We give in Table~\ref{comptimerecapproach} timings for various bilinear maps
and for the implementations of~\myalgorithm{BDEZ} and~\myalgorithm{BDEZStab}. The number of tests represents the number of calls
to~\myalgorithm{HasRankOneBasis}.

It is possible to estimate the time it would take to obtain a result for a bilinear rank problem
out of reach for \myalgorithm{BDEZ} or~\myalgorithm{BDEZStab}.
We denote by $\mathcal{N}_t$ the number of calls to \myalgorithm{HasRankOneBasis} in
these algorithms when the input $r$ is equal to $\ell+t$.
($\ell$ is the dimension of the vector space $\VecSp{T}$ corresponding to the bilinear map).
Since when $r$ is too large, \myalgorithm{BDEZ} is too expensive, there is a practical limit on the known
values of $\mathcal{N}_{t}$, $t$ being a positive integer.
We consider the ratio $\lceil\frac{\mathcal{N}_{t}}
{\mathcal{N}_{t-1}}\rceil$ to estimate $\mathcal{N}_{t+1}$. Assuming that this ratio decreases with $t$, which seems to hold empirically, we have
\begin{equation}
	\label{eq:estimation}
	\mathcal{N}_{t+1} \leq \left\lceil\frac{\mathcal{N}_{t}}{\mathcal{N}_{t-1}}\right\rceil\cdot \mathcal{N}_{t},
\end{equation}
$t$ being a positive integer of $\closedinterval{1}{r-\ell}$.

Thus, we are able to predict timings for bilinear maps indicated in Table~\ref{comptimerecapproach} via
to this assumption, which allows us to compare Algorithm~\myalgorithm{BDEZ} to other approaches for problems
of larger sizes.
We estimate the number of tests by computing
$$\mathcal{N}_{t}\cdot\left\lceil\frac{\mathcal{N}_{t}}{\mathcal{N}_{t-1}}\right\rceil^{r-\ell-t}$$
where $r-\ell$ is the difference $\rk(\VecSp{T}) - \dim(\VecSp{T})$ for $\VecSp{T}$ representing a bilinear map
and $t$ is the largest integer for which we are able to compute $\mathcal{N}_{t}$.
The time can be estimated with a similar technique.
We observe that the speed-up seems to match with $\#\stb(\VecSp{T})$, as expected.
The estimated values in Table~\ref{comptimerecapproach} relying on~\myalgorithm{BDEZStab}
have not been effectively done because
the implementation of~\myalgorithm{CoveringSetsMethod} allowed us to obtain more results, more efficiently.
The estimations rely on the heuristic given by the Inequality~\ref{eq:estimation}.
In the global strategy, we increase progressively the lower bound $r$ on the rank,
before running \myalgorithm{BDEZ}, \myalgorithm{BDEZStab} or \myalgorithm{CoveringSetsMethod}. For $r < \rk(\VecSp{T})$, the time
spent in those algorithms is negligible, because of the exponential growth of their complexity.

\begin{table}
	\begin{center}
	\begin{tabular}{|c|c|c|c|c|}
	\hline
		bilinear map & rank & algorithm & nb. of tests  & time (s)\\
	\hline
	\hline
		\multirow{2}{*}{$\myalgorithm{MatProduct}_{(2,2,2)}$}& \multirow{2}{*}{$7$} & \myalgorithm{BDEZ}      & $1.05 \cdot 10^6$ & $8.5 \cdot 10$ \\
		\cline{3-5} &   &  \myalgorithm{BDEZStab}   & $6.8 \cdot 10^3$ & $5.0 \cdot 10^{-1}$  \\
		\hline
		\multirow{2}{*}{$\myalgorithm{MatProduct}_{(3,2,3)}$} & \multirow{2}{*}{$15$} & \myalgorithm{BDEZ} &$9.2 \cdot 10^{19}$ (est.) &  $1.1 \cdot 10^{17}$ (est.)\\
		\cline{3-5} & &  \myalgorithm{BDEZStab}& $2.6 \cdot 10^{13}$ (est.) & $3.4\cdot 10^{10}$ (est.)\\
		\cline{3-5} & &  \textbf{\myalgorithm{CoveringSetsMethod}}& $\mathbf{1.6\cdot 10^9}$ & $\mathbf{8.5\cdot 10^{5}}$ \textbf{}\\
		\hline
		\multirow{3}{*}{$\myalgorithm{MatProduct}_{(2,3,2)}$} & \multirow{3}{*}{$11$} & \myalgorithm{BDEZ} &$2.3\cdot 10^{23}$ (est.) &  $2.7 \cdot 10^{20}$ (est.)\\
		\cline{3-5} & &  \myalgorithm{BDEZStab}& $4.6\cdot 10^{18}$ (est.) & $5.4\cdot 10^{15}$ (est.)\\
		\cline{3-5} & &  \textbf{\myalgorithm{CoveringSetsMethod}}& $\mathbf{6.3\cdot 10^{10}}$ & $\mathbf{4.1\cdot 10^{6}}$ \textbf{}\\
	\hline
	\hline
	\multirow{2}{*}{$\myalgorithm{ShortProduct}_3$} & \multirow{2}{*}{$5$} & \myalgorithm{BDEZ} &  $5.9\cdot 10^2$ & $1.4 \cdot 10^{-1}$ \\
	\cline{3-5} &   &  \myalgorithm{BDEZStab} & $3.4\cdot 10$ & $0.0$  \\
	\hline
	\multirow{3}{*}{$\myalgorithm{ShortProduct}_4$} & \multirow{3}{*}{$8$} & \myalgorithm{BDEZ} & $5.2\cdot 10^7$  & $4.3\cdot 10^3$  \\
	\cline{3-5} &   &  \myalgorithm{BDEZStab} & $3.1\cdot 10^5$ & $2.7\cdot 10$ \\
	\cline{3-5} &   &  \textbf{\myalgorithm{CoveringSetsMethod} }& $\mathbf{2.8\cdot 10^2}$ & $\mathbf{3.0}$ \textbf{}\\
	\hline
		\multirow{3}{*}{$\myalgorithm{ShortProduct}_{5}$} & \multirow{3}{*}{$11$} & \myalgorithm{BDEZ} &$1.8\cdot 10^{16}$ (est.) &  $5.7 \cdot 10^{12}$ (est.)\\
		\cline{3-5} & &  \myalgorithm{BDEZStab}& $6.9\cdot 10^{11}$ (est.) & $2.2 \cdot 10^8$ (est.) \\
		\cline{3-5} & &  \textbf{\myalgorithm{CoveringSetsMethod}}& $\mathbf{6.3\cdot 10^6}$ & $\mathbf{2.4\cdot 10^{3}}$ \textbf{} \\
	\hline
		\multirow{2}{*}{$\myalgorithm{ShortProduct}_{6}$} & \multirow{2}{*}{$14$} & \myalgorithm{BDEZ} &$3.9\cdot 10^{26}$ (est.)&  $4.7\cdot 10^{23}$ (est.)\\
		\cline{3-5} & &  \myalgorithm{BDEZStab}& $2.0\cdot 10^{19}$ (est.) & $2.7 \cdot 10^{16}$ (est.)\\
	\hline
	\hline
	\multirow{2}{*}{$\myalgorithm{CirculantProduct}_3$} & \multirow{2}{*}{$4$} & \myalgorithm{BDEZ} &  $36$ & $0.0$ \\
	\cline{3-5} &   &  \myalgorithm{BDEZStab} &  $6$ & $0.1 \cdot 10^{-2}$ \\
	\hline
	\multirow{2}{*}{$\myalgorithm{CirculantProduct}_4$} & \multirow{2}{*}{$8$} & \myalgorithm{BDEZ} & $5.2\cdot 10^7$  & $4.3\cdot 10^3$ \\
	\cline{3-5} &   &  \myalgorithm{BDEZStab} & $3.1\cdot 10^5$ & $2.7\cdot 10$ \\
	\hline
		\multirow{3}{*}{$\myalgorithm{CirculantProduct}_{5}$} & \multirow{3}{*}{$10$} & \myalgorithm{BDEZ} &$4.0\cdot 10^{13}$ (est.) & $1.2\cdot 10^{10}$ (est.)\\
		\cline{3-5} & &  \myalgorithm{BDEZStab}& $1.0\cdot 10^{10}$ (est.) & $3.5\cdot 10^{6}$ (est.)\\
		\cline{3-5} & &  \textbf{\myalgorithm{CoveringSetsMethod}}& $ \mathbf{8.8\cdot 10^8}$ & $\mathbf{5.4\cdot 10^{3}}$ \textbf{}\\
	\hline
		\multirow{2}{*}{$\myalgorithm{CirculantProduct}_{6}$} & \multirow{2}{*}{$12$} &\myalgorithm{BDEZ} & $1.0\cdot 10^{20}$ (est.) & $1.3\cdot 10^{17}$ (est.)\\
		\cline{3-5} & &  \myalgorithm{BDEZStab}& $1.1\cdot 10^{15}$ (est.) & $1.5\cdot 10^{12}$ (est.)\\
	\hline
	\end{tabular}
	\end{center}
	\caption{Timings obtained with Algorithm~\myalgorithm{BDEZ} and~\myalgorithm{BDEZStab}
	for various bilinear maps over $K = \mathbb{F}_2$.}
	\label{comptimerecapproach}
\end{table}


It is not clear how to estimate timings for our approach~\myalgorithm{CoveringSetsMethod} beyond what has been done and reported in Table~\ref{comptimerecapproach}. However, for the set
of bilinear maps for which~\myalgorithm{CoveringSetsMethod} allows one to compute all the optimal formulae, we observe
a clear speed-up compared to~\myalgorithm{BDEZStab}.

In order to compute bilinear maps of larger degrees using this method, we need to be able to compute and store all the elements of
$$\quotient{\AL{10}}{\GLGL{{10}}{{10}}}$$
for $\myalgorithm{ShortProduct}_6$ (and even more for other bilinear maps),
which has not been done yet and requires a specific effort for
an optimized implementation of the algorithm described in Section~\ref{sec:combobj}.
Moreover, being able to decompose a matrix product of larger dimensions, such as $3\times 3$ by $3\times 3$,
requires to improve on
the theoretical aspect of our strategy, since the size of the required set
$$\quotient{\AL{9,9,15}}{\GLGL{9}{9}}$$
is expected to be too large, based on the apparent exponential growth of the progression of the sets described in Table~\ref{tab:timeomegad}.

In the following, we describe how we computed optimal formulae for bilinear maps given in Table~\ref{comptimerecapproach} via
our approach using the stems. We provide some technical details, specific to each bilinear map, necessary for an implementation.

\subsection{Matrix product $3\times 2$ by $2\times 3$}
\label{exp:matprod323}
We give in this section the timings obtained with our approach for computing the bilinear rank of the matrix product
$(3,2,3)$ over $\mathbb{F}_2$. We use the same notations as in Section~\ref{sec:apptoprods}.
We recall that
we denote by $\boldsymbol{\Phi}_{3,2,3}$ the bilinear map
$$\begin{matrix}
	\boldsymbol{\Phi}_{3,2,3}:& \mathcal{M}_{3,2}(\mathbb{F}_2)\times \mathcal{M}_{2,3}(\mathbb{F}_2) &\longrightarrow &\mathcal{M}_{3,3}(\mathbb{F}_2)\\
	&(A,B) &\longmapsto& A\cdot B
\end{matrix}.$$
We denote by $\Phi_{i,j}$ the bilinear forms such that
$\Phi_{i,j}(A,B)$ is the coefficient $(i,j)$ of $\boldsymbol{\Phi}_{3,2,3}(A,B)$. The subspace $\VecSp{T}_{3,2,3}$ is defined by
$$\VecSp{T}_{3,2,3} = \spn(\{\Phi_{i,j}\}_{i,j})$$

As described in Section~\ref{sec:computetopoone}, we need to precompute the quotients
$$\quotient{\AL{6+k}}{\GLGL{{6+k}}{{6+k}}}$$
for $k\in \left\{1,2,3\right\}$, and, given the stem that is used, we can restrict the enumeration to
subspaces containing at least one element of rank $6$. The techniques for computing theses subsets are described in Section~\ref{sec:combobj}.

The intermediate sets, corresponding to the quotient $\mathcal{Q}$ computed using \myalgorithm{IntermediateSetViaQuotientComputation}
in Section~\ref{sec:computetopo}, were computed in $1.6\cdot 10^5$ seconds.
They are defined as the following sets:
	$\tilde{\mathcal{E}}_0 = \QCover{7}{\Phi_{0,0}+\Phi_{1,1}+\Phi_{2,2}}$,
	$\tilde{\mathcal{E}}_1 = \QCover{8}{\Phi_{0,0}+\Phi_{1,1},
		\Phi_{0,1}+\Phi_{2,2}}$,
	$\tilde{\mathcal{E}}_2 = \QCover{8}{\Phi_{0,0}+\Phi_{1,1},
		\Phi_{1,1}+\Phi_{2,2}}$,
	$\tilde{\mathcal{E}}_3 = \QCover{8}{\Phi_{0,0}+\Phi_{1,1},\Phi_{2,2}}$,
	$\tilde{\mathcal{E}}_4 = \QCover{9}{\Phi_{0,0},
		\Phi_{1,1}, \Phi_{2,2}}$.
For the set $\tilde{\mathcal{E}}_4$, we actually used an additional trick,
described in~\ref{appendix:hamweightproof}, which allowed us
to consider only a much smaller subset $\tilde{\mathcal{E}}'_4$.

We give in Table~\ref{esttimesols323} the time required to compute the second
step of Section~\ref{sec:computetopoone}, which corresponds to $\mathcal{Q} \circ \mathcal{L}$
calls to \myalgorithm{HasRankOneBasis}.

\begin{table}[H]
	\begin{center}
	\begin{tabular}{|c|c|c|c|c|}
	\hline
	set              & cardinality & nb. tests & time (s) & nb. of solutions found\\
	\hline                                   
	$\tilde{\mathcal{E}}_0$  & $8.8\cdot 10$    & $1.2\cdot 10^8$ & $2.0\cdot 10^5$ & $5$\\
	\hline                              
	$\tilde{\mathcal{E}}_1$  & $7.5\cdot 10^5$  & $2.2\cdot 10^7$ & $3.3\cdot 10^5$ & $13$ \\
	\hline                              
	$\tilde{\mathcal{E}}_2$  & $1.0\cdot 10^4$  & $2.8\cdot 10^5$ & $4.1\cdot 10^2$ & $1$ \\
	\hline                              
	$\tilde{\mathcal{E}}_3$  & $2.7\cdot 10^5$  & $5.9\cdot 10^8$ & $9.1\cdot 10^5$ & $46$ \\
	\hline                              
	$\tilde{\mathcal{E}}'_4$ & $2.5\cdot 10^7$  & $9.1\cdot 10^8$ & $1.3\cdot 10^6 $ & $2$\\
	\hline
	\end{tabular}
	\end{center}
	\caption{Computation of elements of $\protect\CCover{15}{\VecSp{T}_{3,2,3}}$.}
	\label{esttimesols323}
\end{table}

In conclusion, we are able to decompose $\boldsymbol{\Phi}_{3,2,3}$ over $\mathbb{F}_2$ and to
give all the possible optimal decompositions. We have a speed-up of $10^4$ compared to our implementation
of Algorithm~\myalgorithm{BDEZStab}. Although the rank of this bilinear map was already
known thanks to Hopcroft and Kerr~\cite{doi:10.1137/0120004}, determining
all the possible optimal decompositions was not a well studied problem to our knowledge.

We prove with our algorithm that there is only one class of equivalence of
vector spaces $\VecSp{W}\in \AL{6,6,15}$ containing $\VecSp{T}_{3,2,3}$, for the group
action induced by $\stb(\VecSp{T}_{3,2,3})$.
It is interesting to note that this is also the case for $\VecSp{T}_{2,2,2}$.
We do not have this kind of result for the short product for example.
\subsection{Matrix product $2\times 3$ by $3\times 2$}
\label{exp:matprod}
We denote by $\boldsymbol{\Phi}_{2,3,2}$ the bilinear map
$$\begin{matrix}
	\boldsymbol{\Phi}_{2,3,2}:& \mathcal{M}_{2,3}(\mathbb{F}_2)\times \mathcal{M}_{3,2}(\mathbb{F}_2) &\longrightarrow &\mathcal{M}_{2,2}(\mathbb{F}_2)\\
	&(A,B) &\longmapsto& A\cdot B
\end{matrix}$$
and $\Phi_{i,j}$ its coefficients.

We compute the following sets, corresponding to the quotient $\mathcal{Q}$ computed with \myalgorithm{IntermediateSetViaQuotientComputation}
in Section~\ref{sec:computetopo}, within $1.5\cdot 10^6$ seconds:
\begin{itemize}
	\item $\tilde{\mathcal{E}}_0 = \QCover{9}{\Phi_{0,0}+\Phi_{1,1},\Phi_{0,0}+\Phi_{0,1}+\Phi_{1,0}}$,
	\item $\tilde{\mathcal{E}}_1 = \QCover{9}{\Phi_{0,0}+\Phi_{1,1},\Phi_{0,0}}$.
\end{itemize}
We used, in particular, the fact that for any basis $\mathcal{B}$ of $\VecSp{T}_{2,3,2}$, there exist two elements
$\Phi$ and $\Psi$ of $\mathcal{B}$ such that there exists an element in $\spn(\{\Phi,\Psi\})$ whose decomposition over
$(\Phi_{0,0},\Phi_{1,1},\Phi_{0,1},\Phi_{1,0})$
has the following shape:
$$(1,0,\lambda_3,\lambda_4)\ \text{or}\ (0,1,\lambda_3,\lambda_4).$$
The timings for the second step of the method
proposed in Section~\ref{sec:computetopo} are described in Table~\ref{esttimesols232}.

\begin{table}[H]
	\begin{center}
	\begin{tabular}{|c|c|c|c|c|}
	\hline
		set & cardinality & nb. tests & time (s) & nb. of solutions found\\
	\hline
		$\tilde{\mathcal{E}}_0$ & $139$ & $5.0\cdot 10^4$ & $6.2\cdot 10^4$ & $44$\\
	\hline                
		$\tilde{\mathcal{E}}_1$ & $3.8\cdot 10^8 $ & $6.3\cdot 10^{10}$ & $4.1\cdot 10^6$ & $5,614$ \\
	\hline
	\end{tabular}
	\end{center}
	\caption{Computation of $\protect\CCover{11}{\VecSp{T}_{2,3,2}}$.}
	\label{esttimesols232}
\end{table}

We obtained a speed-up of $10^9$ compared to our implementation of~\myalgorithm{BDEZStab}, and we found
$1{,}096{,}452$ elements of $\CCover{11}{\VecSp{T}_{2,3,2}}$, divided in $196$ equivalence classes of solutions with
respect to the action of $\stb(\VecSp{T}_{2,3,2})$.
The computations described in Table~\ref{esttimesols232} used an improved basic test~\myalgorithm{HasRankOneBasis} specialized for $\VecSp{T}_{2,3,2}$.
This test uses the fact that, given a subspace $\VecSp{W}$ of $\tilde{\mathcal{E}}_0$ or $\tilde{\mathcal{E}}_1$, we have two elements $t_0$ and $t_1$ in $\VecSp{T}_{2,3,2}$
such that there exist $w_0,w_1 \in \VecSp{W}$ such that $t_0-w_0$ and $t_1-w_1$ have rank one.
We enumerate the elements $w \in \VecSp{W}$ such that the rank of $t_0-w$ or $t_1-w$ is one, instead of enumerating the whole set of rank-one
bilinear forms.

\subsection{Short product}
\label{exp:shortprod}

We present in this section the timings obtained with our method for the decomposition
of the short product.
We managed to obtain all the elements of $\CCover{r}{\ShProdMap{\ell}}$, where $\ShProdMap{\ell}$ is the vector space
generated by the bilinear forms associated to $\myalgorithm{ShortProduct}_\ell$ for $\ell=4$ and $\ell=5$ and $r = \rk(\ShProdMap{\ell})$.

\begin{table}[H]
	\begin{center}
	\begin{tabular}{|c|c|c|c|c|c|}
	\hline
	bilinear map & nb. of tests & time (s) & nb. of solutions & equivalence classes\\
	\hline
		$\myalgorithm{ShortProduct}_{4}$ & $2.8\cdot 10^2$ & $3.0$ & $1{,}440$ & $220$\\
	\hline
		$\myalgorithm{ShortProduct}_{5}$ & $6.3\cdot 10^6$ & $2.4\cdot 10^3 $ & $146{,}944$ & $11{,}424$\\
	\hline
	\end{tabular}
	\end{center}
	\caption{Computation of $\protect\CCover{r}{\ShProdMap{\ell}}$.}
	\label{esttimeshortprod}
\end{table}

The last column of Table~\ref{esttimeshortprod} describes the number of equivalence classes of
vector spaces in $\CCover{r}{\ShProdMap{\ell}}$, with respect to the group $\stb(\ShProdMap{\ell})$.

\subsection{Circulant product}

We present in this section how to find, with our approach, optimal decompositions of the polynomial product modulo $(X^5-1)$.
We denote by $\VecSp{T}$ the target space spanned by the coefficients $\Phi_i$ of the bilinear map
$$\boldsymbol{\Phi}:(A,B) \mapsto A\cdot B \bmod (X^5-1) = \begin{pmatrix}
	\Phi_0\\
	\Phi_1\\
	\Phi_2\\
	\Phi_3\\
	\Phi_4\\
\end{pmatrix} = 
\begin{pmatrix}
	a_4b_1+a_3b_2+a_2b_3+a_1b_4+a_0b_0\\
	a_4b_2+a_3b_3+a_2b_4+a_1b_0+a_0b_1\\
	a_4b_3+a_3b_4+a_2b_0+a_1b_1+a_0b_2\\
	a_4b_4+a_3b_0+a_2b_1+a_1b_2+a_0b_3\\
	a_4b_0+a_3b_1+a_2b_2+a_1b_3+a_0b_4\\
\end{pmatrix}
.$$
The structure of $\VecSp{T}$ allows us to gain an interesting speed-up. Indeed, $\VecSp{T}$ has the following structure:
there exists, up to a constant mulitplicative factor, a unique element of rank one $\phi = \Phi_0 +\Phi_1+\Phi_2+\Phi_3+\Phi_4$
and a hyperplane $H$ such that $H$ contains all the elements of rank $4$
and such that all the elements of rank $5$ are included in $\spn(\{\phi\}) \oplus H$.
Moreover, the action of $\stb(\VecSp{T})$ on $H-\{0\}$ is transitive (proved by an exhaustive enumeration
in $\mathbb{F}_2$), which
means that all the elements of rank $4$ are in the same orbit.
Consequently, it is also transitive on $\spn(\{\phi\})\oplus H$ and all the elements of rank $5$ are
in the same orbit.

Let $\mathcal{B} = \{\Phi_0,\ldots,\Phi_4\}$ be a basis of $\VecSp{T}$.
We distinguish then $2$ cases: either there exists $i$ such that $\Phi_i$ has rank $5$, or
there is no such $i$, which implies that $\phi \in \mathcal{B}$.
We deduce from these observations the following sets to compute:

\begin{itemize}
	\item $\tilde{\mathcal{E}}_0 = \QCover{6}{\Phi_4}$ and
	\item $\mathcal{E}_1 = \quotient{\Cover{9}{\VecSp{H}}}{\stb(\VecSp{H})}$ (any element $\VecSp{V}\in\mathcal{E}_1$ satisfies
		$\VecSp{T}\subset \VecSp{V}+ \spn(\{\phi\})\in\AL{10}$).
\end{itemize}

We obtain the set $\mathcal{E}_1$ via the computation of a covering of $\mathcal{E}_1$ obtained with
$$\tilde{\mathcal{E}}_1 = \QCover{6}{\Phi_0+\Phi_1+\Phi_2+\Phi_3}.$$

\begin{table}[H]
	\begin{center}
	\begin{tabular}{|c|c|c|c|c|}
	\hline
		set & cardinality & nb. tests & time (s) & nb. of solutions found\\
	\hline
	$\tilde{\mathcal{E}}_0$ &  $5.2\cdot 10$ & $8.7\cdot 10^7$ & $3.1\cdot 10^3$ & $0$\\
	\hline
	$\tilde{\mathcal{E}}_1$ &  $2.0\cdot 10^3$ & $6.7\cdot 10^5$ & $2.4\cdot 10^2$ & $264$ \\
	\hline
	\end{tabular}
	\end{center}
	\caption{Computation of $\protect\CCover{10}{\VecSp{T}}$.}
	\label{esttimesolscirc5}
\end{table}
We have in Table~\ref{esttimesolscirc5} the timings for the second step of the procedure
described in Section~\ref{sec:computetopoone}.
The set $\CCover{10}{\VecSp{T}}$ contains $2025$ elements divided in $9$ equivalence classes of solutions.
Interestingly, the set $\tilde{\mathcal{E}}_0$ does not correspond to any element of $\CCover{10}{\VecSp{T}}$. It
means that, for a basis $\mathcal{B}$ of bilinear forms of rank one containing $\phi$ and generating a subspace of
$\CCover{10}{\VecSp{T}}$, the coordinate of the elements of rank $4$ on $\phi$ is zero.

\section{Conclusions}
One of the most challenging problems in the field of bilinear complexity is the decomposition of the
bilinear map given by the product of $3\times 3$ matrices.
Currently, our approach cannot be used to tackle this problem. However, we believe
that it could be approached by further research in the direction of the Hamming weight idea developed
in~\ref{appendix:hamweightproof}.
An important obstacle is the fact that, assuming that the rank is $21$, it would require to compute
$\quotient{\AL{15}}{\GLGL{9}{9}}$,
which is very large.

Another aspect which is not well understood currently for our approach is to establish a realistic complexity analysis.
It requires a theoretical understanding of how the cardinality of the quotients
$\quotient{\AL{d}}{\GLGL{d}{d}}$
behave asymptotically and a classification of their representatives.

Further research could focus on symmetric decompositions of bilinear maps, which have applications
for the multiplication of polynomials over ``small'' finite fields (such as $\mathbb{F}_2$).
Especially, we can improve on the upper bounds on the rank of the product of two polynomials of
fixed degrees by improving on the bilinear complexity of the multiplication algorithms used in the Chudnovsky-Chudnovsky
approach~\cite{CHUDNOVSKY1988285,Randriambololona2012489,Rambaud2015}.

Finally, the approach proposed in this work allows one to compute exhaustively the optimal formulae for
new bilinear maps, which was not feasible with~\cite{Barbulescu2012}. Moreover, it uses
combinatorial objects which are not well documented in the litterature, which may rekindle curiosity for them.

\section*{Acknowledgements}
The author is grateful to J\'er\'emie Detrey and Emmanuel
Thom\'e for their helpful comments and suggestions.

\section*{Bibliography}
\bibliographystyle{abbrvurl}
\bibliography{karanshort}  
\appendix
\section{Computation of stabilizers}
\subsection{Stabilizer of the short product}
\label{app:stabshort}
In this section, we prove Theorem~\ref{thm:stabshortprod}, using the notations of Section~\ref{subsec:structstabshort}:
the bilinear map
$$\boldsymbol{\Phi}_\ell: (A,B) \longmapsto A\cdot B \bmod X^\ell = 
\begin{pmatrix}
	a_{\ell-1}b_0 + \cdots + a_0b_{\ell-1}\\
	\vdots\\
	a_1b_0 + a_0b_1 \\
	a_0b_0 \\
\end{pmatrix}$$
is the bilinear map corresponding to the short product modulo $X^\ell$, $\Phi_0,\ldots,\Phi_{\ell-1}$
are the bilinear forms such that
$$\forall i \geq 0,\ \Phi_i(A,B) = \sum_{j\in \closedinterval{0}{\ell-1-i}} a_{i-j}b_j$$
and $\ShProdMap{\ell}$ is the subspace $\spn(\{\Phi_0,\ldots,\Phi_{\ell-1}\})$.

We recall that $\ShProdMap{\ell}$ is represented
by the ring $K[N]$ of polynomials of degree less than or equal to $\ell-1$ evaluated in the matrix $N$,
which is a nilpotent matrix.
For example, for $\ell = 4$,
$$a_0N^0 + a_1N^1+a_2N^2 +a_3N^3 = 
\begin{pmatrix}
	a_0 & a_1 & a_2 & a_3 \\
	0 & a_0 & a_1 & a_2 \\
	0 & 0 & a_0 & a_1 \\
	0 & 0 & 0 & a_0 \\
\end{pmatrix}.
$$

We observe that the bilinear forms of rank exactly $\ell$ within $\ShProdMap{\ell}$ are described by the matrices that can be expressed
as $P(N)$ where $P$ is a polynomial of degree smaller or equal to $\ell-1$ over $K$ such that $P(0) \neq 0$.
\stabshortprod*
\begin{proof}
	\begin{itemize}
		\item 
	First, we prove that, for any element $M\in \ShProdMap{\ell}$ of rank $\ell$, there exists $R \in K[N]$ a polynomial
	of degree at most $\ell-1$ such that $R(0) \neq 0$ and $R(N) = M$, and that
	$$\exists M_1 \in \stb(\ShProdMap{\ell}),\ I_\ell \cdot M_1 = R(N).$$


	Any element in the orbit of $I_\ell$ has rank $\ell$ and any element of rank $\ell$ in $\ShProdMap{\ell}$
	is associated to a polynomial $R \in K[N]$ evaluated in $N$ of degree $\ell-1$ such that $R(0) \neq 0$.
	It remains to prove that the orbit of $I_\ell$ corresponds exactly to the set of rank-$\ell$ elements.
	Given $R \in K[N]$ such that $R(0) \neq 0$, we denote by $M_1(R)$ the element $(I_\ell,R(N))$. This element is in $\stb(\ShProdMap{\ell})$ because, for
any $S$, we have $R(N)S(N) = (RS \bmod X^{\ell})(N)$, which is a polynomial
evaluated in $N$ of degree at most $\ell-1$. We have:
	$$I_\ell \cdot M_1(R) = \Trans{(I_\ell)} \cdot I_\ell \cdot R(N) = R(N).$$
		\item
	We prove that, for any element $M\in \ShProdMap{\ell}$ of rank $\ell-1$, there exists $R \in K[N]$ a polynomial
	of degree at most $\ell-1$ such that $R(0) = 0$, $R'(0) \neq 0$ and $R(N) = M$, and that
	$$\exists M_2 \in \stb(I_\ell) \cap \stb(\ShProdMap{\ell}),\ N\cdot M_2= R(N).$$

	An element of the orbit of $N$ is an element of rank $\ell-1$ and an element
	of rank $\ell-1$ in $\ShProdMap{\ell}$ is associated to a polynomial $R \in K[N]$ evaluated in $N$
	of degree at most $\ell-1$ such that $R(0) = 0$ and $R'(0) \neq 0$.
	It remains to prove that $N$ can mapped to any element of rank $\ell-1$ via the action of
	$\stb(I_\ell) \cap \stb(\ShProdMap{\ell})$.


			Let $e_\ell$ be the vector $(0,\cdots,0,1)$, such that $R(N)\cdot e_\ell$ corresponds to
	the last colum of $R(N)$.
	We have $R(N)^{\ell-1} e_\ell \neq 0$. Thus, let $P(N)$ be the matrix whose columns are given by the tuple
	$(R(N)^{\ell-1} \cdot e_\ell, R(N)^{\ell-2} \cdot e_\ell,\ldots,R(N) \cdot e_\ell, e_\ell)$.
			We have $R(N)P(N) = (0, R(N)^{\ell-1} \cdot e_\ell,\ldots,R(N)^2 \cdot e_\ell, R(N)\cdot e_\ell) = R(N)$
			and $P(N)N = (0,R(N)^{\ell-1} \cdot e_\ell, R(N)^{\ell-2} \cdot e_\ell,\ldots,R(N) \cdot e_\ell)$. Consequently,
	we have $R(N) P(N) = P(N)N$ and $P(N)^{-1} R(N) P(N) = N$.
	We take $M_2(R) = (\Trans{P(N)},P(N)^{-1})$: $$N \cdot M_2(R) = R(N)\ \text{and}\ M_2(R) 
	\in \stb(I) \cap \stb(\ShProdMap{\ell}).$$

		\item 
	Let $(\Psi,\Psi')$ be a couple of elements of $\ShProdMap{\ell}$ such that
			$\rk(\Psi) = \ell$ and $\rk(\Psi') = \ell-1$.
	Let $(P,P')$ be the corresponding matrices.
	According to the previous points, there exist $M_1 \in \stb(\ShProdMap{\ell})$ such that
			$I_\ell \cdot M_1 = P$
	and $M_2 \in \stb(\ShProdMap{\ell}) \cap \stb(I_\ell)$ such that
			$N \cdot M_2 = P'\cdot M_1^{-1}$. Consequently, we have
			$$(I_\ell,N) = (P\cdot M_1^{-1} \cdot M_2^{-1},P' \cdot M_1^{-1} \cdot M_2^{-1}).$$

		\item 
	We prove that we have $\stb(I_\ell)\cap \stb(N) \subset \stb(\ShProdMap{\ell})$ and that,
	for any $M_3 \in \stb(I_\ell)\cap \stb(N)$,
	there exists $R \in K[N]$ a polynomial
	of degree at most $\ell-1$ such that $R(0) \neq 0$ and 
	$$M_3= (\Trans{(R(N)^{-1})},R(N)).$$

	Let $M_3 \in \stb(I_\ell)\cap \stb(N)$. Since $M_3 \in \stb(I_\ell)$,
	there exists $P \in \GL_\ell$ such that $M_3 = (\Trans{(P^{-1})},P)$ and,
	since $M_3 \in \stb(N)$, $P^{-1} N P = N$.
	We have $PN = N P$.
	
	Multiplying a matrix by $N$ on the left shifts the rows upward and multiplying $N$ on the right
	shifts the columns on the right.
	Therefore, denoting by $p_{ij}$ the coefficients of $P$, with $p_{00} \neq 0$ and $p_{i0} = 0$ for  $i \geq 1$,
	we have
	$$\forall (i,j)\in \closedinterval{1}{\ell-1}\times \closedinterval{0}{\ell-1}, p_{i,j} = p_{i+1,j+1}.$$
	More particularly, $P$ is equal to the evaluation in $N$ of a polynomial $R$ such that $R(0) \neq 0$, from which we deduce that
	$$M_3 = (\Trans{(R(N)^{-1})},R(N))\ \text{and}\ M_3 \in \stb(\ShProdMap{\ell}).$$

	Given the form of the elements of $\stb(I_\ell)\cap \stb(N)$, its cardinality is equal
	to the number of polynomials $R$ of degree at most $\ell-1$ such that $R(0) \neq 0$, which is $\#K^{\ell-1}(K-1)$.
	Combining with the fact that there are $\#K^{\ell-1}(K-1) \cdot \#K^{\ell-2}(K-1)$ pairs $(\Psi,\Psi')$ of elements of $\ShProdMap{\ell}$
	such that $\rk(\Psi) = \ell$ and $\rk(\Psi') = \ell-1$, we have $\#\stb(\ShProdMap{\ell}) = \#K^{3\ell-4}(\#K-1)^3$.
	\end{itemize}
\end{proof}
\subsection{Stabilizer of the matrix product}
\label{app:stabmat}
	We denote by $\VecSp{T}_{p,q,r}$ the vector space given by the product of matrices $p\times q$ by
	$q\times r$, which is isomorphic to $\mathcal{M}_{p,r} \otimes I_q$ (we do not use the canonical basis for this
	representation).
	For the group action $M  \cdot(X,Y) \mapsto \Trans{X} M Y$, we want to prove Theorem~\ref{thm:stabmat}.
	\stabmat*
	\begin{proof}
	Let $(X,Y)$ be a pair of invertible matrices such that $\Trans{X} \VecSp{T}_{p,q,r} Y = \VecSp{T}_{p,q,r}$.
	For any $i \in \closedinterval{0}{p-1}$ and $j \in \closedinterval{0}{q-1}$, we denote by
	$M_{i,j}$ the matrix $\Trans{X} \cdot (e_{i,j}) \cdot Y$, where $e_{i,j}$ is the canonical basis of $\mathcal{M}_{p,r}$.
	Denoting by $X_{i,h}$ the $q\times q$ blocks of $X$ and $Y_{\ell,j}$ the $q\times q$ blocks of $Y$, we have
	$M_{i,j} = (X_{i,h}Y_{j,\ell})_{h,\ell}$ for any $i$ and $j$. Consequently, since $\Trans{X} \cdot (e_{i,j}) \cdot Y \in \VecSp{T}_{p,q,r}$,
	we have
	\begin{equation}
		\label{eq:equalidentity}
		\forall i,j,h,\ell,\ X_{i,h}Y_{j,\ell} \in \spn(\{I_q\}).
	\end{equation}

	Let $(i,h)$ such that $X_{i,h}$ is not null and $j$ any integer in $\closedinterval{0}{q-1}$.
	We have  the inclusion
		$$X_{i,h} \cdot \spn(\closedinterval{Y_{j,0}}{Y_{j,q-1}}) \subset \spn(\{I_q\})$$
	and, since $Y$ is invertible, we even have the equality.
	Thus, for any $(i,h)$ such that $X_{i,h}$ is not null, we have shown that $X_{i,h}$ is invertible. We have the same property for the blocks
	of $Y$.

	Combining the fact that the blocks of $X$ and $Y$ that are not null are invertible and Equation~\eqref{eq:equalidentity}, we
	can conclude that the stabilizer of $\VecSp{T}_{p,q,r}$ is generated by matrices $(X,Y)$
	such that there exists $g \in \GL_q$ satisfying
	$$\Trans{X} \in \GL_p \otimes g \text{ and } Y \in \GL_r \otimes g^{-1}.$$
	\end{proof}

\section{Using the Hamming weight for the matrix product}
\label{appendix:hamweightproof}
We describe in this section a trick allowing one to speed-up the execution of our approach for the matrix product. However,
this part is technical and can be skipped on a first read.

We still denote by $\VecSp{T}$ the subspace of $\Bl{6}{6}$ corresponding to the coefficients of the product of
$3\times 2$ by $2 \times 3$ matrices.
We recall the stem of $\VecSp{T}$ that we consider:
	$$\mathcal{C} = \{\spn(\{\Phi_{0,0}+\Phi_{1,1}+\Phi_{2,2}\}),\spn(\{\Phi_{0,0}+\Phi_{1,1},\Phi_{0,1}+\Phi_{2,2}\}),
	\spn(\{\Phi_{0,0}+\Phi_{1,1},\Phi_{1,1}+\Phi_{2,2}\}),$$
	$$\spn(\{\Phi_{0,0}+\Phi_{1,1},\Phi_{2,2}\}),\spn(\{\Phi_{0,0},\Phi_{1,1},\Phi_{2,2}\})\}.$$
We define the following sets:
\begin{itemize}
	\item $\mathcal{E}_0 = \Cover{7}{\Phi_{0,0}+\Phi_{1,1}+\Phi_{2,2}}$,
	\item $\mathcal{E}_1 = \Cover{8}{\Phi_{0,0}+\Phi_{1,1},\Phi_{0,1}+\Phi_{2,2}}$,
	\item $\mathcal{E}_2 = \Cover{8}{\Phi_{0,0}+\Phi_{1,1},\Phi_{1,1}+\Phi_{2,2}}$,
	\item $\mathcal{E}_3 = \Cover{8}{\Phi_{0,0}+\Phi_{1,1},\Phi_{2,2}}$ and
	\item $\mathcal{E}_4 = \Cover{9}{\Phi_{0,0},\Phi_{1,1},\Phi_{2,2}}$.
\end{itemize}
In theory, we have to enumerate the elements of the sets
$\QCover{7}{\Phi_{0,0}+\Phi_{1,1}+\Phi_{2,2}}$, $\QCover{8}{\Phi_{0,0}+\Phi_{1,1},\Phi_{0,1}+\Phi_{2,2}}$,
$\QCover{8}{\Phi_{0,0}+\Phi_{1,1},\Phi_{1,1}+\Phi_{2,2}}$, $\QCover{8}{\Phi_{0,0}+\Phi_{1,1},\Phi_{2,2}}$ and
$\QCover{9}{\Phi_{0,0},\Phi_{1,1},\Phi_{2,2}}$, denoted by
$\tilde{\mathcal{E}}_0$, $\tilde{\mathcal{E}}_1$, $\tilde{\mathcal{E}}_2$, $\tilde{\mathcal{E}}_3$ and $\tilde{\mathcal{E}}_4$, respectively.
However, one can notice that, given $\VecSp{V} \in \CCover{15}{\VecSp{T}}$ such that
$$\exists \VecSp{W} \in \tilde{\mathcal{E}}_4, \sigma \in \stb(\{\Phi_{0,0},\Phi_{1,1},\Phi_{2,2}\}),\ 
\VecSp{V} = \VecSp{T}+ \VecSp{W}\circ \sigma,$$
it may happen that there exists $\VecSp{W}'\subset\VecSp{V}$ such that
$$\VecSp{V} = \VecSp{T}+ \VecSp{W}'$$
and
$$
\exists \sigma \in \stb(\VecSp{T}),\ 
\VecSp{W}'\circ \sigma \in
\begin{cases}
	\Cover{7}{\Phi_{0,0}+\Phi_{1,1}+\Phi_{2,2}}\\
	\hfil \text{or}\\
	\Cover{8}{\Phi_{0,0}+\Phi_{1,1},\Phi_{0,1}+\Phi_{2,2}}\\
	\hfil \text{or}\\
	\Cover{8}{\Phi_{0,0}+\Phi_{1,1},\Phi_{1,1}+\Phi_{2,2}}\\
	\hfil \text{or}\\
	\Cover{8}{\Phi_{0,0}+\Phi_{1,1},\Phi_{2,2}}\\
\end{cases}.$$
In other terms, the $\VecSp{V}$'s corresponding to the$5$ sets to enumerate do not form a partition
of $\CCover{15}{\VecSp{T}}$.

Thus, we propose, if possible, to enumerate a subset of $\mathcal{E}_4$,
rather than the whole set, without losing exhaustivity.
The strategy that is proposed is related to the notion of Hamming weight of the elements
$\Phi_{0,0}$, $\Phi_{1,1}$ and $\Phi_{2,2}$.

\begin{definition}[Hamming weight for $\AL{d}$]
	Let $\VecSp{W}\in \AL{d}$ and $\mathcal{B}= (\psi_0,\ldots,\psi_{d-1})$ a basis of rank-one bilinear forms of $\VecSp{W}$.
	Any $x \in \VecSp{W}$ has a unique decomposition over $\mathcal{B}$:
	$$x = \sum_{0\leq t<d} \lambda_t \cdot \psi_t.$$
	We define its Hamming weight over $\mathcal{B}$ as
	$$\mathbb{H}_{\mathcal{B}}(x) = \#\{ t \in \closedinterval{0}{d-1} \ |\ \lambda_t \neq 0\}.$$
	We can extend the definition of the Hamming weight to any subset $\mathcal{S}$ of $\VecSp{W}$:
	$$\mathbb{H}_{\mathcal{B}}(\mathcal{S}) = \min(\{\#I \ |\ I \subset \closedinterval{0}{d-1}, \mathcal{S} \subset
	\spn(\{\psi_i\}_{i\in I})\}).$$
\end{definition}
The Hamming weight over some basis has a useful property related to the bilinear rank stated in Lemma~\ref{prop:rankhw}.
\begin{lemma}
	\label{prop:rankhw}
	Let $\VecSp{W}\in \AL{d}$ and $\mathcal{B}$ a basis of $\VecSp{W}$ composed of rank-one bilinear forms.
	For any subset $\mathcal{S}$ of $\VecSp{W}$, we have
	$$\rk(\spn(\mathcal{S})) \leq \mathbb{H}_{\mathcal{B}}(\mathcal{S}).$$
\end{lemma}
\begin{proof}
	Clear from the definition of the rank of a set $\mathcal{S}$ given in Definition~\ref{def:ranksp}.
\end{proof}

We describe in Theorem~\ref{thm:mainthmhamweight} what is the subset of $\tilde{\mathcal{E}}_4$
that we consider.
\begin{theorem}
\label{thm:mainthmhamweight}
	Let $\VecSp{W}$ be a subspace such that
	$\VecSp{W} \in \tilde{\mathcal{E}}_4\ \text{and}\ \VecSp{T}_{3,2,3}+\VecSp{W} \in \AL{15}$
	and let $\mathcal{B}$ be a basis of $\VecSp{W}$ composed of rank-one bilinear forms.
	Let $\tilde{\mathcal{E}}'$ be the subset of elements $\VecSp{W} \in \tilde{\mathcal{E}}_4$
	such that
	$$\mathbb{H}_{\mathcal{B}}(\Phi_{0,0}+\Phi_{1,1}+\Phi_{2,2})=
	\mathbb{H}_{\mathcal{B}}(\Phi_{0,0})+\mathbb{H}_{\mathcal{B}}(\Phi_{1,1})+\mathbb{H}_{\mathcal{B}}(\Phi_{2,2})
	\ \text{and}\ \mathbb{H}_{\mathcal{B}}(\Phi_{0,0}+\Phi_{1,1}+\Phi_{2,2}) > 6.$$
	We obtain all the elements of $\CCover{15}{\VecSp{T}}$ via the enumeration of
	$\tilde{\mathcal{E}}_{0}$, $\tilde{\mathcal{E}}_{1}$, $\tilde{\mathcal{E}}_{2}$, $\tilde{\mathcal{E}}_{3}$ and $\tilde{\mathcal{E}}'$.
\end{theorem}
We prove Theorem~\ref{thm:mainthmhamweight} within $2$ steps:
\begin{enumerate}
	\item We prove in Lemma~\ref{lem:linham} that if $\mathbb{H}_{\mathcal{B}}(\Phi_{0,0}+\Phi_{1,1}+\Phi_{2,2})\neq
	\mathbb{H}_{\mathcal{B}}(\Phi_{0,0})+\mathbb{H}_{\mathcal{B}}(\Phi_{1,1})+\mathbb{H}_{\mathcal{B}}(\Phi_{2,2})$,
	a subspace $\VecSp{V}$ obtained as $\VecSp{V} = \VecSp{T} + \VecSp{W}\circ \sigma$ can also be obtained
	as $\VecSp{V} = \VecSp{T} + \VecSp{W}'\circ \sigma$, with $\VecSp{W}' \in \tilde{\mathcal{E}}_{0}$, $\tilde{\mathcal{E}}_{1}$ or $\tilde{\mathcal{E}}_{3}$.
	\item Otherwise, if $\mathbb{H}_{\mathcal{B}}(\Phi_{0,0}+\Phi_{1,1}+\Phi_{2,2}) =
	\mathbb{H}_{\mathcal{B}}(\Phi_{0,0})+\mathbb{H}_{\mathcal{B}}(\Phi_{1,1})+\mathbb{H}_{\mathcal{B}}(\Phi_{2,2})$,
	it remains to prove that we do not lose in generality if we assume that
	$\mathbb{H}_{\mathcal{B}}(\Phi_{0,0}+\Phi_{1,1}+\Phi_{2,2}) > 6$, which is done in Lemma~\ref{lem:onelementrank3}.
\end{enumerate}

\begin{lemma}
	\label{lem:linham}
	Let $\VecSp{W} \in \AL{9}$ and let $\VecSp{V} = \VecSp{T}+\VecSp{W}$.
	We assume that $\VecSp{T}+\VecSp{W} \in \CCover{15}{\VecSp{T}}$ and
	$\spn(\{\Phi_{0,0},\Phi_{1,1},\Phi_{2,2}\}) \subset \VecSp{W}$.
	Let $\mathcal{B}$ be a basis of rank-one bilinear forms of $\VecSp{W}$.
	If $\mathbb{H}_{\mathcal{B}}(\Phi_{0,0}+\Phi_{1,1}+\Phi_{2,2})\neq
	\mathbb{H}_{\mathcal{B}}(\Phi_{0,0})+\mathbb{H}_{\mathcal{B}}(\Phi_{1,1})+\mathbb{H}_{\mathcal{B}}(\Phi_{2,2})$,
	there exists $\VecSp{W}'\subset \VecSp{W}$ such that $\VecSp{V} = \VecSp{T} + \VecSp{W}'$ and
	there exists $\sigma' \in \stb(\VecSp{T})$ such that
	$\VecSp{W}'\circ \sigma' \in {\mathcal{E}}_0, {\mathcal{E}}_1$ or ${\mathcal{E}}_3$.
\end{lemma}
\begin{proof}
	We have by hypothesis
	$\mathbb{H}_{\mathcal{B}}(\Phi_{0,0}+\Phi_{1,1}+\Phi_{2,2}) <
\mathbb{H}_{\mathcal{B}}(\Phi_{0,0})+\mathbb{H}_{\mathcal{B}}(\Phi_{1,1})+\mathbb{H}_{\mathcal{B}}(\Phi_{2,2})$.
	Thus, there exist two elements $\Psi\in\mathcal{B}$ and $\Phi\in\{\Phi_{0,0}, \Phi_{1,1}, \Phi_{2,2}\}$ such that
	the coordinate of $\Phi$ on $\Psi$ is not zero and the coordinates of $\Phi_{0,0}+\Phi_{1,1}+\Phi_{2,2}$ on $\Psi$ is zero.
	By considering the vector space $\VecSp{W}' = \spn(\mathcal{B}-\{\Psi\})$, we have
	$\VecSp{W}' \in \AL{8}$.
	Moreover, we have $\VecSp{W} = \spn(\{\Psi\}) \oplus \VecSp{W}' = \spn(\{\Phi\}) \oplus \VecSp{W}'$
	and $\VecSp{T}+(\spn(\{\Phi\}) \oplus \VecSp{W}') = \subset\VecSp{T} + \VecSp{W}'$.
	Thus, $\dim(\VecSp{T}+\VecSp{W}') = \dim(\VecSp{T}+\VecSp{W}) = 15$.
	Consequently, $\dim(\VecSp{T} \cap \VecSp{W}') = 2$.
	
	If there exists in $\VecSp{T} \cap \VecSp{W}'$ two elements $\Phi_1$ and $\Phi_2$ of rank smaller or equal to $4$ 
	such that $\Phi_1 + \Phi_2 = \Phi_{0,0}+\Phi_{1,1}+\Phi_{2,2}$, then
	$$\exists \sigma \in \stb(\VecSp{T}),\ \VecSp{W}'\circ\sigma \in
	\begin{cases}
		\Cover{8}{\Phi_{0,0}+\Phi_{1,1},\Phi_{0,1}+\Phi_{2,2}}\\
		\hfil \text{or}\\
		\Cover{8}{\Phi_{0,0}+\Phi_{1,1},\Phi_{2,2}}\\
	\end{cases}.$$
	Otherwise, there exists $\VecSp{W}''\subset\VecSp{W}'$ such that
	$$\exists \sigma \in \stb(\VecSp{T}),\ \VecSp{W}''\circ\sigma \in
		\Cover{7}{\Phi_{0,0}+\Phi_{1,1}+\Phi_{2,2}}$$
	and
	$\VecSp{T}+\VecSp{W}'' \in \AL{15}$,
	which concludes.
\end{proof}

\begin{lemma}
	\label{lem:onelementrank3}
	Let $\VecSp{V} \in \CCover{15}{\VecSp{T}}$. The subspace $\VecSp{V}$ satisfies hypotheses \textbf{H1} and
	\textbf{H2} state as follows.
	\begin{itemize}
		\item[\textbf{H1:}] For any $\VecSp{W} \subset \VecSp{V}$ such that
	there exists $\sigma \in \stb(\VecSp{T})$ satisfying $\VecSp{W}\circ \sigma \in \Cover{9}{\Phi_{0,0},\Phi_{1,1},\Phi_{2,2}}$ and $\VecSp{T}+\VecSp{W} \in \AL{15}$,
	we have, for any basis $\mathcal{B}$ of rank-one bilinear forms of $\VecSp{W}$,
	$$\mathbb{H}_{\mathcal{B}}(\Phi_{0,0}+\Phi_{1,1}+\Phi_{2,2})=
	\mathbb{H}_{\mathcal{B}}(\Phi_{0,0})+\mathbb{H}_{\mathcal{B}}(\Phi_{1,1})+\mathbb{H}_{\mathcal{B}}(\Phi_{2,2}).$$
		\item[\textbf{H2:}] There do not exist $\VecSp{W}\subset\VecSp{V}$ and $\sigma \in \stb(\VecSp{T})$
	such that $\VecSp{W}\circ \sigma \in {\mathcal{E}}_0, {\mathcal{E}}_1$, $\mathcal{E}_2$ or ${\mathcal{E}}_3$ and $\VecSp{V} = \VecSp{T} + \VecSp{W}$ (in other terms, $\VecSp{V}$ can not be obtained via the enumeration of
	$\tilde{\mathcal{E}}_0$, $\tilde{\mathcal{E}}_1$, $\tilde{\mathcal{E}}_2$ or $\tilde{\mathcal{E}}_3$).
	\end{itemize}

	Then, there exists $\VecSp{W}' \subset \VecSp{V}$ and $\sigma' \in \stb(\VecSp{T})$
	such that $\VecSp{W}'\circ \sigma' \in \Cover{9}{\Phi_{0,0},\Phi_{1,1},\Phi_{2,2}}$, $\VecSp{T} + \VecSp{W}' \in \AL{15}$,
	and $\VecSp{W}'$ has a basis $\mathcal{B}'$ of rank-one
	bilinear forms such that
	$$\mathbb{H}_{\mathcal{B}'\circ \sigma'}(\Phi_{0,0}+\Phi_{1,1}+\Phi_{2,2}) >6.$$
\end{lemma}
\begin{proof}
	Let $\VecSp{W}\in \AL{6}$ be such that
	$\VecSp{T}\oplus\VecSp{W} \in \AL{15}$.
	Take a basis $\mathcal{W}$ of $\VecSp{W}$, and complete it into a basis $\mathcal{B}$
	of $\VecSp{T}\oplus \VecSp{W}$ using $9$ rank-one bilinear forms, denoted
	by $\{\psi_i\}_{0\leq i<9}$. For all $i \in \closedinterval{0}{8}$,
	write $\psi_i = \Phi_i + \Psi_i$, with $\Phi_i\in \VecSp{T}$ and $\Psi_i \in \VecSp{W}$.
	The $\Phi_i$'s form a basis of $\VecSp{T}$.
	In our context, we are concerned by the case $\rk(\Phi_i) = 2$ for any $i$ (otherwise \textbf{H2} is not satisfied).
	Since we assume Hypothesis \textbf{H1}, it is enough to prove that there exists $i$ such that
	$\mathbb{H}_{\mathcal{B}}(\Phi_i) > 2$.

	There is necessarily a couple $(\Phi_i,\Phi_{i'})$ such that
	$\mathbb{H}_{\mathcal{B}}(\Phi_i+\Phi_{i'})<\mathbb{H}_{\mathcal{B}}(\Phi_i)+\mathbb{H}_{\mathcal{B}}(\Phi_{i'}).$
	Otherwise, we would have $\#\mathcal{B} \geq 2\cdot 9 = 18 \neq 15 = \rk(\VecSp{T})$.
	
	If $\mathbb{H}_{\mathcal{B}}(\Phi_i) = 2$ and $\mathbb{H}_{\mathcal{B}}(\Phi_{i'}) = 2$,
	then
	$$\mathbb{H}_{\mathcal{B}}(\{\Phi_i,\Phi_{i'}\})\leq 3$$
	and, by Lemma~\ref{prop:rankhw}, $\rk(\{\Phi_i,\Phi_{i'}\}) \leq 3$.
	Henceforth, we prove that this is contradictory, because $\rk(\{\Phi_i,\Phi_{i'}\}) = 4$. 
	Indeed, there are two cases.
	\begin{itemize}
	\item If $\Phi_i + \Phi_{i'}$ has rank $4$, the conclusion follows.
	\item If $\spn(\{\Phi_i,\Phi_{i'}\})$ is isomorphic to $\VecSp{T}_{2,2,1}$, whose rank is equal to $4$:
	$\VecSp{T}_{2,2,1}$ and $\VecSp{T}_{2,1,2}$ have the same rank according to~\cite{deGroote:1987:LCB:23720}
	and $\VecSp{T}_{2,1,2}$ is a vector space of dimension one generated by a bilinear form of rank $4$.
	\end{itemize}
	Consequently, $\mathbb{H}_{\mathcal{B}}(\Phi_i) > 2$ or $\mathbb{H}_{\mathcal{B}}(\Phi_{i'}) > 2$.
\end{proof}
\end{document}